\documentclass[journal]{IEEEtran}
\usepackage{cite}
\usepackage{amsmath,amssymb,amsfonts,bm}
\usepackage[linesnumbered,ruled,vlined]{algorithm2e}
\usepackage{mathrsfs}
\usepackage{amsthm}
\usepackage{dsfont}
\usepackage{algorithmic}
\usepackage{subfigure}
\usepackage{graphicx}
\usepackage{textcomp}
\usepackage{threeparttable}
\usepackage{xcolor}
\usepackage{bbm}
\usepackage{makecell}
\usepackage[nolist,withpage]{acronym}
\usepackage{array}
\usepackage{tikz}
\usetikzlibrary{spy}
\usepackage{pgfplots}
\pgfplotsset{compat=newest}

\definecolor{copperrose}{rgb}{0.6, 0.4, 0.4}
\definecolor{azure}{rgb}{0.0, 0.5, 1.0}
\definecolor{ashgrey}{rgb}{0.7, 0.75, 0.71}
\definecolor{chestnut}{rgb}{0.8, 0.36, 0.36}
\definecolor{airforceblue}{rgb}{0.36, 0.54, 0.66}
\definecolor{cadmiumorange}{rgb}{0.93, 0.53, 0.18}
\definecolor{bleudefrance}{rgb}{0.19, 0.55, 0.91}
\definecolor{carolinablue}{rgb}{0.6, 0.73, 0.89}
\definecolor{blue(ncs)}{rgb}{0.0, 0.53, 0.74}
\definecolor{dodgerblue}{rgb}{0.12, 0.56, 1.0}
\definecolor{cssgreen}{rgb}{0.0, 0.5, 0.0}
\definecolor{cadmiumgreen}{rgb}{0.0, 0.42, 0.24}
\definecolor{cadmiumorange}{rgb}{0.93, 0.53, 0.18}
\definecolor{amaranth}{rgb}{0.9, 0.17, 0.31}
\definecolor{bluegray}{rgb}{0.4, 0.6, 0.8}
\definecolor{cerulean}{rgb}{0.0, 0.48, 0.65}
\definecolor{ceil}{rgb}{0.57, 0.63, 0.81}

\usepackage{etoolbox}
\makeatletter
\newif\if@in@acrolist
\AtBeginEnvironment{acronym}{\@in@acrolisttrue}
\newrobustcmd{\LU}[2]{\if@in@acrolist#1\else#2\fi}

\newcommand{\ACF}[1]{{\@in@acrolisttrue\acf{#1}}}

\begin{document}


\begin{acronym}[LTE-Advanced]
  \acro{2G}{Second Generation}
  \acro{3-DAP}{3-Dimensional Assignment Problem}
  \acro{3G}{3$^\text{rd}$~Generation}
  \acro{3GPP}{3$^\text{rd}$~Generation Partnership Project}
  \acro{4G}{4$^\text{th}$~Generation}
  \acro{5G}{5$^\text{th}$~Generation}
  \acro{AA}{Antenna Array}
  \acro{AC}{Admission Control}
  \acro{ACS}{Add-Compare-Select}
  \acro{AD}{Attack-Decay}
  \acro{ADC}{analog-to-digital converter}
  \acro{ADMM}{alternating direction method of multipliers}
  \acro{ADSL}{Asymmetric Digital Subscriber Line}
  \acro{AHW}{Alternate Hop-and-Wait}
  \acro{AI}{Artificial Intelligence}
  \acro{AirComp}{over-the-air computation}
  \acro{AMC}{Adaptive Modulation and Coding}
  \acro{ANN}{artificial neural network}
  \acro{AP}{\LU{A}{a}ccess \LU{P}{p}oint}
  \acro{APA}{Adaptive Power Allocation}
  \acro{ARMA}{Autoregressive Moving Average}
  \acro{ARQ}{\LU{A}{a}utomatic \LU{R}{r}epeat \LU{R}{r}equest}
  \acro{ATES}{Adaptive Throughput-based Efficiency-Satisfaction Trade-Off}
  \acro{AWGN}{additive white Gaussian noise}
  \acro{BAA}{\LU{B}{b}roadband \LU{A}{a}nalog \LU{A}{a}ggregation}
  \acro{BB}{Branch and Bound}
  \acro{BCD}{block coordinate descent}
  \acro{BCJR}{Bahl-Cocke-Jelinek-Raviv}
  \acro{BD}{Block Diagonalization}
  \acro{BER}{Bit Error Rate}
  \acro{BF}{Best Fit}
  \acro{BFD}{bidirectional full duplex}
  \acro{BLER}{BLock Error Rate}
  \acro{BPC}{Binary Power Control}
  \acro{BPSK}{Binary Phase-Shift Keying}
  \acro{BRA}{Balanced Random Allocation}
  \acro{BS}{base station}
  \acro{BSUM}{block successive upper-bound minimization}
  \acro{CAP}{Combinatorial Allocation Problem}
  \acro{CAPEX}{Capital Expenditure}
  \acro{CBF}{Coordinated Beamforming}
  \acro{CBR}{Constant Bit Rate}
  \acro{CBS}{Class Based Scheduling}
  \acro{CC}{Congestion Control}
  \acro{CCCP}{constrained convex-concave procedure}
  \acro{CDF}{Cumulative Distribution Function}
  \acro{CDMA}{Code-Division Multiple Access}
  \acro{CE}{\LU{C}{c}hannel \LU{E}{e}stimation}
  \acro{CL}{Closed Loop}
  \acro{CLPC}{Closed Loop Power Control}
  \acro{CML}{centralized machine learning}
  \acro{CNR}{Channel-to-Noise Ratio}
  \acro{CNN}{\LU{C}{c}onvolutional \LU{N}{n}eural \LU{N}{n}etwork}
  \acro{CP}{computation point}
  \acro{CPA}{Cellular Protection Algorithm}
  \acro{CPICH}{Common Pilot Channel}
  \acro{CoCoA}{\LU{C}{c}ommunication efficient distributed dual \LU{C}{c}oordinate \LU{A}{a}scent}
  \acro{CoMAC}{\LU{C}{c}omputation over \LU{M}{m}ultiple-\LU{A}{a}ccess \LU{C}{c}hannels}
  \acro{CoMP}{Coordinated Multi-Point}
  \acro{CQI}{Channel Quality Indicator}
  \acro{CRM}{Constrained Rate Maximization}
	\acro{CRN}{Cognitive Radio Network}
  \acro{CS}{Coordinated Scheduling}
  \acro{CSI}{\LU{C}{c}hannel \LU{S}{s}tate \LU{I}{i}nformation}
  \acro{CSMA}{\LU{C}{c}arrier \LU{S}{s}ense \LU{M}{m}ultiple \LU{A}{a}ccess}
  \acro{CUE}{Cellular User Equipment}
  \acro{D2D}{device-to-device}
  \acro{DAC}{digital-to-analog converter}
  \acro{DC}{difference of convex}
  \acro{DCA}{Dynamic Channel Allocation}
  \acro{DE}{Differential Evolution}
  \acro{DFT}{Discrete Fourier Transform}
  \acro{DIST}{Distance}
  \acro{DL}{downlink}
  \acro{DMA}{Double Moving Average}
  \acro{DML}{Distributed Machine Learning}
  \acro{DMRS}{demodulation reference signal}
  \acro{D2DM}{D2D Mode}
  \acro{DMS}{D2D Mode Selection}
  \acro{DPC}{Dirty Paper Coding}
  \acro{DRA}{Dynamic Resource Assignment}
  \acro{DSA}{Dynamic Spectrum Access}
  \acro{DSGD}{\LU{D}{d}istributed \LU{S}{s}tochastic \LU{G}{g}radient \LU{D}{d}escent}
  \acro{DSM}{Delay-based Satisfaction Maximization}
  \acro{ECC}{Electronic Communications Committee}
  \acro{EFLC}{Error Feedback Based Load Control}
  \acro{EI}{Efficiency Indicator}
  \acro{eNB}{Evolved Node B}
  \acro{EPA}{Equal Power Allocation}
  \acro{EPC}{Evolved Packet Core}
  \acro{EPS}{Evolved Packet System}
  \acro{E-UTRAN}{Evolved Universal Terrestrial Radio Access Network}
  \acro{ES}{Exhaustive Search}
  \acro{FC}{\LU{F}{f}usion \LU{C}{c}enter}
  \acro{FD}{\LU{F}{f}ederated \LU{D}{d}istillation}
  \acro{FDD}{frequency divisionov duplex}
  \acro{FDM}{Frequency Division Multiplexing}
  \acro{FDMA}{\LU{F}{f}requency \LU{D}{d}ivision \LU{M}{m}ultiple \LU{A}{a}ccess}
  \acro{FedAvg}{\LU{F}{f}ederated \LU{A}{a}veraging}
  \acro{FER}{Frame Erasure Rate}
  \acro{FF}{Fast Fading}
  \acro{FL}{federated learning}
  \acro{FSB}{Fixed Switched Beamforming}
  \acro{FSJD}{full-state joint decoder}
  \acro{FSK}{Frequency-Shift Keying}
  \acro{FST}{Fixed SNR Target}
  \acro{FTP}{File Transfer Protocol}
  \acro{GA}{Genetic Algorithm}
  \acro{GBR}{Guaranteed Bit Rate}
  \acro{GLR}{Gain to Leakage Ratio}
  \acro{GOS}{Generated Orthogonal Sequence}
  \acro{GPL}{GNU General Public License}
  \acro{GRP}{Grouping}
  \acro{HARQ}{Hybrid Automatic Repeat Request}
  \acro{HiCoMAC}{Histogram-State Coded Multiple Access Computing}
  \acro{HD}{half-duplex}
  \acro{HMS}{Harmonic Mode Selection}
  \acro{HOL}{Head Of Line}
  \acro{HSDPA}{High-Speed Downlink Packet Access}
  \acro{HSPA}{High Speed Packet Access}
  \acro{HTTP}{HyperText Transfer Protocol}
  \acro{ICMP}{Internet Control Message Protocol}
  \acro{ICI}{Intercell Interference}
  \acro{ID}{Identification}
  \acro{IETF}{Internet Engineering Task Force}
  \acro{ILP}{Integer Linear Program}
  \acro{JRAPAP}{Joint RB Assignment and Power Allocation Problem}
  \acro{UID}{Unique Identification}
  \acro{IABP}{importance-adaptive bit-partitioning}
  \acro{IID}{\LU{I}{i}ndependent and \LU{I}{i}dentically \LU{D}{d}istributed}
  \acro{IIR}{Infinite Impulse Response}
  \acro{ILP}{Integer Linear Problem}
  \acro{IMT}{International Mobile Telecommunications}
  \acro{INV}{Inverted Norm-based Grouping}
  \acro{IoT}{Internet of Things}
  \acro{IP}{Integer Programming}
  \acro{IPv6}{Internet Protocol Version 6}
  \acro{ISD}{Inter-Site Distance}
  \acro{ISI}{Inter Symbol Interference}
  \acro{ITU}{International Telecommunication Union}
  \acro{JAFM}{joint assignment and fairness maximization}
  \acro{JAFMA}{joint assignment and fairness maximization algorithm}
  \acro{JOAS}{Joint Opportunistic Assignment and Scheduling}
  \acro{JOS}{Joint Opportunistic Scheduling}
  \acro{JP}{Joint Processing}
	\acro{JS}{Jump-Stay}
  \acro{KKT}{Karush-Kuhn-Tucker}
  \acro{L3}{Layer-3}
  \acro{LAC}{Link Admission Control}
  \acro{LA}{Link Adaptation}
  \acro{LC}{Load Control}
  \acro{LDC}{\LU{L}{l}earning-\LU{D}{d}riven \LU{C}{c}ommunication}
  \acro{LDPC}{low-density parity-check}
  \acro{LOS}{line of sight}
  \acro{LP}{Linear Programming}
  \acro{LSB}{least significant bit}
  \acro{LTE}{Long Term Evolution}
	\acro{LTE-A}{\ac{LTE}-Advanced}
  \acro{LTE-Advanced}{Long Term Evolution Advanced}
  \acro{LRA}{Low Rank Approximation}
  \acro{M2M}{Machine-to-Machine}
  \acro{MAC}{multiple access channel}
  \acro{MAP}{maximum a posterior}
  \acro{MANET}{Mobile Ad hoc Network}
  \acro{MC}{Modular Clock}
  \acro{MCS}{Modulation and Coding Scheme}
  \acro{MDB}{Measured Delay Based}
  \acro{MDI}{Minimum D2D Interference}
  \acro{MF}{Matched Filter}
  \acro{MG}{Maximum Gain}
  \acro{MH}{Multi-Hop}
  \acro{MIMO}{\LU{M}{m}ultiple \LU{I}{i}nput \LU{M}{m}ultiple \LU{O}{o}utput}
  \acro{MINLP}{mixed integer nonlinear programming}
  \acro{MIP}{mixed integer programming}
  \acro{MISO}{multiple input single output}
  \acro{ML}{machine learning}
  \acro{MLE}{maximum likelihood estimator}
  \acro{MLWDF}{Modified Largest Weighted Delay First}
  \acro{MME}{Mobility Management Entity}
  \acro{MMSE}{minimum mean squared error}
  \acro{MOS}{Mean Opinion Score}
  \acro{MPF}{Multicarrier Proportional Fair}
  \acro{MRA}{Maximum Rate Allocation}
  \acro{MR}{Maximum Rate}
  \acro{MRC}{Maximum Ratio Combining}
  \acro{MRT}{Maximum Ratio Transmission}
  \acro{MRUS}{Maximum Rate with User Satisfaction}
  \acro{MS}{Mode Selection}
  \acro{MSB}{most significant bit}
  \acro{MSE}{\LU{M}{m}ean \LU{S}{s}quared \LU{E}{e}rror}
  \acro{MSI}{Multi-Stream Interference}
  \acro{MTC}{Machine-Type Communication}
  \acro{MTSI}{Multimedia Telephony Services over IMS}
  \acro{MTSM}{Modified Throughput-based Satisfaction Maximization}
  \acro{MU-MIMO}{Multi-User Multiple Input Multiple Output}
  \acro{MU}{Multi-User}
  \acro{NAS}{Non-Access Stratum}
  \acro{NB}{Node B}
	\acro{NCL}{Neighbor Cell List}
  \acro{NLP}{Nonlinear Programming}
  \acro{NLOS}{non-line of sight}
  \acro{NMSE}{normalized mean square error}
  \acro{NOMA}{\LU{N}{n}on-\LU{O}{o}rthogonal \LU{M}{m}ultiple \LU{A}{a}ccess}
  \acro{NORM}{Normalized Projection-based Grouping}
  \acro{NP}{non-polynomial time}
  \acro{NRT}{Non-Real Time}
  \acro{NSPS}{National Security and Public Safety Services}
  \acro{O2I}{Outdoor to Indoor}
  \acro{OFDMA}{\LU{O}{o}rthogonal \LU{F}{f}requency \LU{D}{d}ivision \LU{M}{m}ultiple \LU{A}{a}ccess}
  \acro{OFDM}{Orthogonal Frequency Division Multiplexing}
  \acro{OFPC}{Open Loop with Fractional Path Loss Compensation}
	\acro{O2I}{Outdoor-to-Indoor}
  \acro{OL}{Open Loop}
  \acro{OLPC}{Open-Loop Power Control}
  \acro{OL-PC}{Open-Loop Power Control}
  \acro{OPEX}{Operational Expenditure}
  \acro{ORB}{Orthogonal Random Beamforming}
  \acro{JO-PF}{Joint Opportunistic Proportional Fair}
  \acro{OSI}{Open Systems Interconnection}
  \acro{PAIR}{D2D Pair Gain-based Grouping}
  \acro{PAM}{pulse amplitude modulation}
  \acro{PAPR}{Peak-to-Average Power Ratio}
  \acro{P2P}{Peer-to-Peer}
  \acro{PC}{Power Control}
  \acro{PCI}{Physical Cell ID}
  \acro{PDCCH}{physical downlink control channel}
  \acro{PDD}{penalty dual decomposition}
  \acro{PDF}{Probability Density Function}
  \acro{PER}{Packet Error Rate}
  \acro{PF}{Proportional Fair}
  \acro{P-GW}{Packet Data Network Gateway}
  \acro{PL}{Pathloss}
  \acro{RLT}{reformulation linearization technique}
  \acro{PRB}{Physical Resource Block}
  \acro{PROJ}{Projection-based Grouping}
  \acro{ProSe}{Proximity Services}
  \acro{PS}{\LU{P}{p}arameter \LU{S}{s}erver}
  \acro{PSO}{Particle Swarm Optimization}
  \acro{PUCCH}{physical uplink control channel}
  \acro{PZF}{Projected Zero-Forcing}
  \acro{QAM}{quadrature amplitude modulation}
  \acro{QoS}{quality of service}
  \acro{QPSK}{quadrature phase shift keying}
  \acro{QCQP}{quadratically constrained quadratic programming}
  \acro{RAISES}{Reallocation-based Assignment for Improved Spectral Efficiency and Satisfaction}
  \acro{RAN}{Radio Access Network}
  \acro{RA}{Resource Allocation}
  \acro{RAT}{Radio Access Technology}
  \acro{RATE}{Rate-based}
  \acro{RB}{resource block}
  \acro{RBG}{Resource Block Group}
  \acro{REF}{Reference Grouping}
  \acro{ReLU}{rectified linear unit}
  \acro{ReMAC}{repetition for multiple access computing}
  \acro{RF}{radio frequency}
  \acro{RLC}{Radio Link Control}
  \acro{RM}{Rate Maximization}
  \acro{RNC}{Radio Network Controller}
  \acro{RND}{Random Grouping}
  \acro{RRA}{Radio Resource Allocation}
  \acro{RRM}{\LU{R}{r}adio \LU{R}{r}esource \LU{M}{m}anagement}
  \acro{RSCP}{Received Signal Code Power}
  \acro{RSRP}{reference signal receive power}
  \acro{RSRQ}{Reference Signal Receive Quality}
  \acro{RR}{Round Robin}
  \acro{RRC}{Radio Resource Control}
  \acro{RSSI}{received signal strength indicator}
  \acro{RT}{Real Time}
  \acro{RU}{Resource Unit}
  \acro{RUNE}{RUdimentary Network Emulator}
  \acro{RV}{Random Variable}
  \acro{SA}{simulated annealing}
  \acro{SAA}{Small Argument Approximation}
  \acro{SAC}{Session Admission Control}
  \acro{SCM}{Spatial Channel Model}
  \acro{SC-FDMA}{Single Carrier - Frequency Division Multiple Access}
  \acro{SD}{Soft Dropping}
  \acro{S-D}{Source-Destination}
  \acro{SDPC}{Soft Dropping Power Control}
  \acro{SDMA}{Space-Division Multiple Access}
  \acro{SDR}{semidefinite relaxation}
  \acro{SDP}{semidefinite programming}
  \acro{SeMAC}{sequential modulation for AirComp}
  \acro{SeMAC-PA}{SeMAC with power adaptation}
  \acro{SER}{Symbol Error Rate}
  \acro{SES}{Simple Exponential Smoothing}
  \acro{S-GW}{Serving Gateway}
  \acro{SGD}{\LU{S}{s}tochastic \LU{G}{g}radient \LU{D}{d}escent}  
  \acro{SINR}{signal-to-interference-plus-noise ratio}
  \acro{SI}{self-interference}
  \acro{SIP}{Session Initiation Protocol}
  \acro{SISO}{\LU{S}{s}ingle \LU{I}{i}nput \LU{S}{s}ingle \LU{O}{o}utput}
  \acro{SIMO}{Single Input Multiple Output}
  \acro{SIR}{Signal to Interference Ratio}
  \acro{SLNR}{Signal-to-Leakage-plus-Noise Ratio}
  \acro{SMA}{Simple Moving Average}
  \acro{SNR}{\LU{S}{s}ignal-to-\LU{N}{n}oise \LU{R}{r}atio}
  \acro{SOCP}{second-order cone programming}
  \acro{SORA}{Satisfaction Oriented Resource Allocation}
  \acro{SORA-NRT}{Satisfaction-Oriented Resource Allocation for Non-Real Time Services}
  \acro{SORA-RT}{Satisfaction-Oriented Resource Allocation for Real Time Services}
  \acro{SPF}{Single-Carrier Proportional Fair}
  \acro{SRA}{Sequential Removal Algorithm}
  \acro{SRS}{sounding reference signal}
  \acro{SU-MIMO}{Single-User Multiple Input Multiple Output}
  \acro{SU}{Single-User}
  \acro{SVD}{singular value decomposition}
  \acro{SVM}{\LU{S}{s}upport \LU{V}{v}ector \LU{M}{m}achine}
  \acro{TCP}{Transmission Control Protocol}
  \acro{TDD}{time division duplex}
  \acro{TDMA}{\LU{T}{t}ime \LU{D}{d}ivision \LU{M}{m}ultiple \LU{A}{a}ccess}
  \acro{TNFD}{three node full duplex}
  \acro{TETRA}{Terrestrial Trunked Radio}
  \acro{TP}{Transmit Power}
  \acro{TPC}{Transmit Power Control}
  \acro{TTI}{transmission time interval}
  \acro{TTR}{Time-To-Rendezvous}
  \acro{TSM}{Throughput-based Satisfaction Maximization}
  \acro{TU}{Typical Urban}
  \acro{UBP}{uniform bit-partitioning}
  \acro{UE}{\LU{U}{u}ser \LU{E}{e}quipment}
  \acro{UEPS}{Urgency and Efficiency-based Packet Scheduling}
  \acro{UL}{uplink}
  \acro{UMTS}{Universal Mobile Telecommunications System}
  \acro{URI}{Uniform Resource Identifier}
  \acro{URM}{Unconstrained Rate Maximization}
  \acro{VR}{Virtual Resource}
  \acro{VoIP}{Voice over IP}
  \acro{WAN}{Wireless Access Network}
  \acro{WCDMA}{Wideband Code Division Multiple Access}
  \acro{WF}{Water-filling}
  \acro{WiMAX}{Worldwide Interoperability for Microwave Access}
  \acro{WINNER}{Wireless World Initiative New Radio}
  \acro{WLAN}{Wireless Local Area Network}
  \acro{WMMSE}{weighted minimum mean square error}
  \acro{WMPF}{Weighted Multicarrier Proportional Fair}
  \acro{WPF}{Weighted Proportional Fair}
  \acro{WSN}{Wireless Sensor Network}
  \acro{WWW}{World Wide Web}
  \acro{XIXO}{(Single or Multiple) Input (Single or Multiple) Output}
  \acro{ZF}{Zero-Forcing}
  \acro{ZMCSCG}{Zero Mean Circularly Symmetric Complex Gaussian}
\end{acronym}

\title{HiCoMAC: \\ Histogram-State Coded Multiple Access Computing \\ via
Computation-Oriented Modulation Design
}

\author{Xiaojing Yan, Carlo Fischione\\
	\normalsize School of Electrical Engineering and Computer Science, KTH Royal Institute of Technology, Stockholm, Sweden\\
	\normalsize Email: \{xiay, carlofi\}@kth.se
}

\newtheorem{theorem}{Theorem}
\newtheorem{prop}{Proposition}
\newtheorem{defi}{Definition}
\newtheorem{lem}{Lemma}
\newtheorem{rem}{Remark}

\maketitle
\thispagestyle{empty}
\pagestyle{empty}

\begin{abstract}
Nowadays, \ac{AirComp} exploits the waveform superposition property of the
wireless multiple access channel to directly compute a function of
distributed data from simultaneously transmitted signals. In this paper, we propose
\ac{HiCoMAC}, a fundamentally new convolutional-coded method for digital \ac{AirComp}
that directly recovers the arithmetic-sum sequence from the
superposition of coded and modulated user signals. \ac{HiCoMAC}
exploits the permutation invariance of the sum function to represent
the multi-user convolutional encoder state by the numbers of users
occupying the individual encoder states. This central idea enables log-max \ac{MAP} histogram-state Viterbi decoding with substantially reduced state
complexity compared to joint full-state decoding. We further introduce the computational free distance to
distinguish histogram paths that produce different arithmetic-sum
sequences and develop a weighted product graph for its shortest path
evaluation. This distance is then used to optimize the rectangular
ratio of a fixed labeled \ac{QAM} constellation to improve the computational accuracy. Simulation results show that the computation-oriented modulation design increases
the computational free distance and reduces the \ac{NMSE}, while
\ac{HiCoMAC} outperforms the considered multi-symbol digital
\ac{AirComp} baselines under equal transmission and energy budgets.
\end{abstract}

\section{Introduction}
\label{sec:introduction}

\ac{AirComp} exploits the waveform
superposition property of a wireless multiple-access channel to
compute a function of distributed data directly from simultaneously
transmitted signals~\cite{csahin2023survey}. Unlike conventional multiple access
communication, where the receiver first decodes the individual
messages and subsequently performs the desired computation~\cite{yin2006ofdma}, AirComp
integrates communication and computation at the physical layer~\cite{goldenbaum2013harnessing}. This property makes AirComp attractive for distributed sensing, learning, and control applications that require efficient aggregation of data generated by a large number of wireless devices
\cite{zhao2022broadband,amiri2020federated,zhu2018mimo,daei2025timely}.

Most early AirComp methods rely on analog transmission, where the
desired aggregate is directly represented by the superposition of
uncoded waveforms~\cite{goldenbaum2014nomographic}. Although this approach provides low communication
latency, its computation accuracy is sensitive to channel noise and fading. Moreover, uncoded analog
transmission is not directly compatible with the digital modulation and channel coding mechanisms used in practical communication
systems~\cite{csahin2022over}. These limitations have motivated increasing interest in
digital AirComp, where the user data are quantized, digitally
modulated, and recovered from a finite set of aggregated signals~\cite{razavikia2023computing}.

\subsection{Related Work and Motivation}

Recently, digital \ac{AirComp} has been investigated from the perspective of modulation design. Particularly, ChannelComp designs digital constellations such that aggregated signals associated with different function outputs remain distinguishable at the receiver \cite{saeed2023ChannelComp}. After such seminal work, to further improve computation reliability, several
multi-symbol transmission methods have subsequently been developed. In particular, \ac{ReMAC} jointly designs the constellation diagram and coded repetition patterns over multiple time slots, thereby introducing transmission redundancy while avoiding destructive overlaps between aggregate signal sequences \cite{yan2025remac}. Bit-slicing divides a quantized value into multiple bit segments and transmits them using separate modulation symbols \cite{liu2025digital}. We extended the approach we proposed in~\cite{yan2025remac} by introducing \ac{SeMAC} that designs multi-symbol transmission sequences to enlarge the separation between different function outputs, with power allocation providing further performance improvement~\cite{yan2025multi}. Bit-partitioning methods divide the quantized bits into groups and jointly optimize their allocation and modulation mapping across multiple symbols~\cite{yan2025joint}. Although these methods exploit multiple channel uses, they do not incorporate the memory of a finite-state channel code or the sequential structure induced by channel encoding.

Channel coding has increasingly been incorporated into traditional multiple access
transmission to introduce redundancy and improve reliability
\cite{ozccelikkale2014short}. The broadband digital AirComp system
in~\cite{you2023broadband} develops joint channel decoding and
aggregation for convolutional and \ac{LDPC} codes, enabling accurate model aggregation from superimposed coded transmissions. For convolutional coding, its full-state joint decoder searches a trellis whose states jointly specify the encoder states of the individual users and recovers the sum after identifying the corresponding joint source sequence. Its reduced-state decoder retains only a metric-selected subset of these joint states at each decoding stage. Although this reduces the decoding complexity, the underlying state representation still distinguishes the individual users, and the reduction is obtained through path pruning, which may degrade decoding accuracy. The study in~\cite{santos2025over} partially overlaps independently encoded user transmissions and performs joint maximum likelihood (ML) decoding with a single Viterbi execution, but recovers the individual user information blocks rather than the computation of a desired function.

Other coded AirComp constructions have also been investigated. The
method in~\cite{xie2026joint} combines non-binary \ac{LDPC} coding
with lattice-based modulation and treats the aligned multi-user
transmission as a virtual point-to-point link for decoding. The method in~\cite{weng2026channel} applies an identical
linear encoding matrix at all users to preserve the aggregation
structure and obtain coding gains against channel distortion. However, existing coded \ac{AirComp} methods either retain user-specific decoding states or focus primarily on reliable coded aggregation, without exploiting the permutation invariance of the desired computation to construct the decoding state space. Moreover, the signal space separation between coded superpositions associated with different computation outputs has not been explicitly incorporated into modulation design. Addressing these issues requires a computation-oriented exploration in digital \ac{AirComp} that connects the coding structure, distance analysis, and constellation design. 

\subsection{Contributions}

In this paper, we propose \ac{HiCoMAC}, a convolutional-coded digital
\ac{AirComp} method for reliable arithmetic-sum computation. In
HiCoMAC, all users employ the same convolutional encoder and labeled
\ac{QAM} constellation, while the receiver directly decodes the
arithmetic-sum sequence from the superimposed signals. \ac{HiCoMAC} exploits the permutation
invariance of the sum function to aggregate joint full-state
encoder states into histogram states. This representation supports Viterbi decoding of the arithmetic-sum sequence while avoiding
the exponential state growth of the joint full-state trellis. We
further develop a computation-oriented distance characterization and
use it to design the modulation constellation. Finally, we evaluate the reliability performance of  \ac{HiCoMAC} via comprehensive simulation results.

The main contributions are summarized as follows:
\begin{itemize}

    \item We propose \ac{HiCoMAC}, a convolutional-coded digital \ac{AirComp} method under an ideally aligned \ac{AWGN} \ac{MAC}. We develop a histogram-state representation that compactly describes the joint multi-user convolutional encoding process without distinguishing the identities of individual users.

    \item Based on the histogram-state representation, we develop a log-max \ac{MAP} Viterbi decoder that directly recovers the arithmetic-sum sequence while reducing the decoding complexity. We show that, for a fixed encoder memory, the number of decoder states grows polynomially with the number of users, instead of exponentially as in the full-state joint Viterbi decoder in~\cite{you2023broadband}.

    \item We introduce the computational free distance to quantify the minimum signal space separation between irreducible computational error events associated with different arithmetic-sum sequences. We construct a modified product graph that represents pairs of histogram paths and prove that the computational free distance can be obtained by the shortest path search in this graph.

    \item We use the computational free distance as the modulation design criterion and optimize the rectangular ratio of a fixed labeled \ac{QAM} constellation. The resulting optimization is solved by a coarse-to-fine search over the ratio, with the computational free distance evaluated efficiently by the Dijkstra algorithm.

    \item Simulation results demonstrate that the proposed \ac{HiCoMAC} design
achieves a larger computational free distance and lower \ac{NMSE} than the corresponding square \ac{QAM} design, and also achieves lower \ac{NMSE} than the considered multi-symbol AirComp strategies based on modulation and repetition design.
\end{itemize}



\section{HiCoMAC System Model}
\label{sec:system_model}

\input{Figures/system_model}

This section presents the system model of \ac{HiCoMAC}. We first describe its convolutional coding and modulation model~\cite{proakis2001digital},
and then construct the histogram-state representation of the joint multi-user convolutional encoding process.

\subsection{Convolutional Coding and Modulation}

We consider a network with \(K\) users and one computation point (CP), where each
user \(k\) has an input value
\(x_k\in[x_{\min},x_{\max}]\). Before transmission, \(x_k\) is
uniformly quantized into a \(B\)-bit sequence $\bm b_k
=[b_k[1],\ldots,b_k[B]]\in\{0,1\}^{B}$, where the quantization spacing is denoted as 
\begin{equation}
    \Delta_B
    =
    \frac{x_{\max}-x_{\min}}{2^B-1}.
    \label{eq:quantization_step}
\end{equation}

The \ac{CP} aims to compute the arithmetic sum
\(x_{\Sigma}=\sum_{k=1}^{K}x_k\). The corresponding sum of the
quantized inputs is
\begin{align}
    x_{\Sigma}^{\mathrm Q}
    &=
    Kx_{\min}
    +
    \Delta_B
    \sum\nolimits_{\tau=1}^{B}
    2^{B-\tau}
    \sum\nolimits_{k=1}^{K}b_k[\tau],
    \label{eq:quantized_sum_expansion}
\end{align}
which is determined by the bit wise arithmetic sums. Thus, we define the desired arithmetic-sum sequence as
\(\bm S=[S[1],\ldots,S[B]]\in\{0,\ldots,K\}^{B}\), where $S[\tau]=\sum\nolimits_{k=1}^{K} b_k[\tau]$.

The bit information block of each user is then encoded using the same \((1,L,m)\) convolutional code, where $L$
and $m$ are the number of output bits and memory order. Hence, the convolutional encoder has \(N=2^m\) states indexed by
\(\mathcal S_{\mathrm c}=\{1,\ldots,N\}\) and the coding rate can be expressed as $R=1/L$.  Here, we consider that all encoders are terminated in the all-zero state.
For each input value, the encoder input length is
\(T=B+m\) with $m$ zero tail bits $b_k[\tau]=0$ for $\tau\in \{B+1,\ldots,T\}$.
When encoding, each bit \(b_k[\tau]\) is processed in one encoder transition, where user $k$ feeds $b_k[\tau]$ into the convolutional encoder and
outputs an $L$-bit coded word
$\bm c_k[\tau]\in\{0,1\}^L$. Let
\(s_k[\tau]\in\mathcal S_{\mathrm c}\) denote the encoder state of
user \(k\) after processing the first \(\tau\) bit inputs, with
\(s_k[0]=1\) representing the all-zero initial state. Then, the next encoder state and coded output can be described as
\begin{equation}
\begin{cases}
    s_k[\tau]
    =
    \delta(s_k[\tau-1],b_k[\tau]),\\
    \bm c_k[\tau]
    =
    \gamma(s_k[\tau-1],b_k[\tau]),
\end{cases}
\end{equation}
where
$\delta:\mathcal S_{\mathrm c}\times\{0,1\}\rightarrow
\mathcal S_{\mathrm c}$ and
$\gamma:\mathcal S_{\mathrm c}\times\{0,1\}\rightarrow\{0,1\}^{L}$.

The channel coded word is then mapped to a complex coded modulation symbol.
Let
$\tilde{\bm{x}}=[\tilde{x}_1,\ldots,\tilde{x}_M]^{\mathsf T}
\in\mathbb C^M$ denote a modulation vector containing a fixed labeled
\ac{QAM} constellation pattern, and $M=2^L$ is the modulation order. The labeling rule is specified by
$\lambda:\{0,1\}^{L}\rightarrow\{1,\ldots,M\}$\footnote{We choose \ac{QAM} for its widespread practical use and structured
representation, which provides a natural and
tractable basis for this initial study of computation-oriented coded
modulation design in digital \ac{AirComp}.}. For each constellation
index $\ell$, we define
$\tilde{x}_\ell=\tilde{x}_{\ell,\mathrm I}+j\tilde{x}_{\ell,\mathrm Q}$,
where $\tilde{x}_{\ell,\mathrm I}=\Re\{\tilde{x}_\ell\}$ and
$\tilde{x}_{\ell,\mathrm Q}=\Im\{\tilde{x}_\ell\}$.
The average energies of its in-phase and quadrature components are represented as
\begin{equation} \nonumber
    E_{\mathrm I}
    =
    \frac{1}{T}
    \sum_{\tau=1}^{T}
    \mathbb E \Big[\widetilde x_{\lambda(\bm c_k[\tau]),I}^{2}\Big],
    ~
    E_{\mathrm Q}
    =
    \frac{1}{T}
    \sum_{\tau=1}^{T}
    \mathbb E \Big[\widetilde x_{\lambda(\bm c_k[\tau]),Q}^{2}\Big],
    \label{eq:axis_energies}
\end{equation}
where the expectation is taken over the equiprobable $B$ information bits, and the $T$ transitions include zero-tail terminations.
In this paper, we propose to adjust the
relative in-phase and quadrature scaling of the modulation via a rectangular ratio
\(\rho>0\). Thus, the resulting normalized constellation point is presented as
\begin{equation}
    \tilde x_\ell(\rho)
    =
    \xi(\rho)
    \left(
    \rho\widetilde x_{\ell,\mathrm I}
    +
    \mathrm j\widetilde x_{\ell,\mathrm Q}
    \right),
    ~\forall \ell \in \{1,\ldots,M\},
    \label{eq:normalized_rectangular_qam}
\end{equation}
where
\begin{equation}
    \xi(\rho)
    =
    \sqrt{
    \frac{E_{\mathrm s}}
    {\rho^2E_{\mathrm I}+E_{\mathrm Q}}
    }
    \label{eq:constellation_normalization_factor}
\end{equation}
ensures an average transmitted coded symbol energy \(E_{\mathrm s}\) over the complete terminated block. In particular,
the case \(\rho=1\) preserves the relative scaling of the original
QAM pattern, whereas \(\rho>1\) increases the in-phase spacing
relative to the quadrature spacing and \(\rho<1\) decreases it.
Accordingly, the coded modulation symbol transmitted by user \(k\) for information bit \(b_k[\tau]\) is given by
\begin{equation}
    \vec x_k[\tau]
    =
    \mu\bigl(\boldsymbol c_k[\tau];\rho\bigr)
    = \tilde x_{\lambda(\boldsymbol c_k[\tau])}(\rho),
    \label{eq:transmitted_coded_symbol}
\end{equation}
where $\mu:\{0,1\}^{L}\rightarrow \mathbb C$ is the modulation encoder.

Given the novelty of the system model, in this paper we consider ideal channel conditions where the user signals are perfectly aligned at the \ac{CP} and the received superposition is affected only by \ac{AWGN}\footnote{Note that we assume accurate \ac{CSI} and perfect time and phase synchronization among the users. Under fading channels, ideal channel inversion can be realized as $p_k[\tau]=h_k^*[\tau]/|h_k[\tau]|^2$ where $h_k[\tau]\in\mathbb C$ denote the corresponding channel coefficient from user $k$ to the \ac{CP}. In deep fading, practical implementations such as truncated channel inversion and pre-equalization can be employed. However, the associated channel estimation, synchronization and power adaptive control are outside the scope of this work and will deserve a future dedicated journal paper.}. Accordingly, the received signal for input $b[\tau]$ is given by
\begin{equation}
    \vec y[\tau]
    =
    \sum\nolimits_{k=1}^{K}\vec{x}_k[\tau]
    +
    \vec z[\tau],
\end{equation}
where $\vec z[\tau]\sim\mathcal{CN}(0,N_0)$ denotes the Gaussian noise. Collecting all
received symbols gives the received sequence
$\bm y=[\vec y[1],\ldots,\vec y[T]]$.
Rather than recovering the individual quantized bit sequences
\(\boldsymbol b_1,\ldots,\boldsymbol b_K\), the \ac{CP}
directly estimates the arithmetic-sum sequence $\widehat{\bm S}=[\widehat S[1],\ldots,\widehat S[B]]$ from \(\bm y\) by a Viterbi decoder, and reconstructs the desired sum as
\begin{equation}
    \widehat x_{\Sigma}
    =
    Kx_{\min}
    +
    \Delta_B
    \sum\nolimits_{\tau=1}^{B}
    2^{B-\tau}\widehat S[\tau].
    \label{eq:reconstructed_computation_sum}
\end{equation}
Specifically, quantization introduces the irreversible distortion
\(\epsilon_k=x_k-x_k^{\mathrm Q}\), where $x_k^{\mathrm Q}$ denotes the quantized input. When
\(\widehat{\boldsymbol S}=\boldsymbol S\), the residual error caused by quantization is given as
\begin{equation}
x_{\Sigma}-\widehat x_{\Sigma}
=\sum\nolimits_{k=1}^{K}\epsilon_k,
\label{eq:residual_computational_error}
\end{equation}
 which determines the
computation error floor.
The overall \ac{HiCoMAC} system model is
illustrated in Fig.~\ref{fig:system_model}.

\input{Figures/state_transition_diagram}

\noindent\textit{Illustrative example.}
To make the coding and modulation procedure concrete, we consider a
two-user setting with memory order \(m=1\) and \(L=2\) coded bits per
input bit. The convolutional encoder therefore has \(N=2\) states, and we index these two states by
\(\sigma\in\{1,2\}\). For this memory-one encoder, we define \(o(\sigma)\) as the memory
content associated with state index \(\sigma\), with \(o(1)=0\) and
\(o(2)=1\). We use the rate-\(1/2\) convolutional encoder with generator pair
\(g(D)=(3,1)_8\). For an input bit \(b\in\{0,1\}\), the next state index is
\begin{equation}
\delta(\sigma,b)=
    \begin{cases}
    1, & b=0,\\
    2, & b=1.
    \end{cases}    
\end{equation}
The coded output is $\gamma(\sigma,b)=(c_1,c_2)$,
where $c_1=b\oplus o(\sigma),~c_2=b$. 
For illustration, we set \(\rho=1\) and unit average energy \(E_{\mathrm s}=1\), yielding
the normalized 4\ac{QAM} mapping
\begin{equation}
    \mu(\bm c;1)
    =
    \begin{cases}
        (1+\mathrm j)/\sqrt{2}, & \bm c=00,\\
        (-1+\mathrm j)/\sqrt{2}, & \bm c=01,\\
        (-1-\mathrm j)/\sqrt{2}, & \bm c=11,\\
        (1-\mathrm j)/\sqrt{2}, & \bm c=10.
    \end{cases}
\end{equation}


\subsection{Histogram-state Representation}

The convolutional encoder can be represented as a finite-state sequential
machine. In this subsection, we first describe the corresponding
single-user state transition graph. 
Then, for the multi-user \ac{AirComp}
setting, to reduce the complexity, we construct a
histogram-state transition multigraph that exploits the permutation invariance of the arithmetic-sum function.

\subsubsection{Single-User State Transition Graph}

We first describe the state-transition graph induced by the single-user
convolutional encoder. Each vertex represents one encoder state index
\(\sigma\in\mathcal S_{\mathrm c}\). Since the encoder input is binary, each state index has two
outgoing transitions, corresponding to \(b=0\) and \(b=1\).

For an input bit \(b\in\{0,1\}\), the encoder moves from state index
\(\sigma\) to the next state index \(\delta(\sigma,b)\) and produces the
coded output \(\gamma(\sigma,b)\in\{0,1\}^L\). The corresponding
transmitted coded modulation symbol is
\(\mu(\gamma(\sigma,b);\rho)\). 
For the illustrative two-user example with \(g(D)=(3,1)_8\) and
\(\rho=1\), the resulting state transition graph is shown in
Fig.~\ref{fig:single_user_state_transition}.

\subsubsection{Histogram-State Transition Multigraph}
\label{subsubsec:histogram_state_multigraph}

In the \(K\)-user \ac{AirComp} setting, since each user has \(N\) possible encoder states, the ordered joint
state space contains \(N^K=2^{mK}\) states. Moreover, every ordered
joint state has \(2^K\) outgoing transitions.
Consequently, both the number of states and outgoing
transitions grow exponentially with the number of users.
However, the \ac{CP} only needs to recover the arithmetic-sum sequence \(\bm S\),
which is invariant to permutations of user identities. Hence, ordered
joint states that differ only by a permutation of their user indices
are equivalent for sum function computation. This permutation
invariance motivates the histogram-state representation~\cite{scott2010histogram}.

\begin{defi}[Histogram state]
\label{def:histogram_state}
After processing the first \(\tau\) encoder inputs, the histogram state is
$\bm n[\tau]=\big[n_1[\tau],\ldots,n_N[\tau]\big]$,
where
\begin{equation}
    n_\sigma[\tau]
    =
    \left|
    \left\{
    k:s_k[\tau]=\sigma
    \right\}
    \right|,~
    \sigma\in\mathcal S_{\mathrm c},
    \label{eq:histogram_state_component}
\end{equation}
is the number of users whose encoder is in state \(\sigma\).
\end{defi}

Since every user occupies exactly one encoder state, the histogram-state space is therefore presented as
\begin{equation}
\label{eq:histogram_state_space}
    \mathcal N_K
    =
    \left\{
    \bm n\in\mathbb Z_{\geq 0}^{N}:
    \sum\nolimits_{\sigma\in\mathcal S_{\mathrm c}}n_\sigma=K
    \right\},
\end{equation}
whose cardinality is $Q_{\mathrm H}
    =|\mathcal N_K|=\binom{K+N-1}{N-1}$.
We enumerate
\(\mathcal N_K=\{\bm n^{(1)},\ldots,\bm n^{(Q_{\mathrm H})}\}\),
where \(\bm n^{(i)}\) denotes the \(\rm i\)-th possible histogram state.
Different multi-user inputs may induce the same pair of source and
target histogram states, resulting in parallel transitions. We therefore
represent the state transition structure by the directed
multigraph $\mathcal G_{\mathrm H}
    =
    \big(
    \mathcal V_{\mathrm H},
    \mathcal E_{\mathrm H}
    \big)$,
where
\(
\mathcal V_{\mathrm H}
=
\{v_1,\ldots,v_{Q_{\mathrm H}}\}
\)
and vertex \(v_i\) corresponds to \(\bm n^{(i)}\).

For a transition from \(v_i\) to \(v_j\), let
\(\bm a_{ij}^{r}
=[a_{ij,1}^{r},\ldots,a_{ij,N}^{r}]\)
denote an input-count vector, where \(a_{ij,\sigma}^{r}\)
is the number of users in encoder state \(\sigma\) whose input
bit is \(1\). The feasible input-count vectors associated with
the vertex pair \((v_i,v_j)\) are collected in
\begin{equation}
\label{eq:input_count_transition_set}
\mathcal A_{ij}
=
\left\{
\bm a_{ij}^{r}\in\mathbb Z_{\geq0}^{N}
\;\middle|\;
\begin{aligned}
&0\leq a_{ij,\sigma}^{r}\leq n_{\sigma}^{(i)},
&&\forall\sigma\in\mathcal S_{\mathrm c},\\
&n_{\sigma'}^{(j)}=H_{ij,\sigma'}^{r},
&&\forall\sigma'\in\mathcal S_{\mathrm c}
\end{aligned}
\right\},
\end{equation}
where
\begin{equation}
\label{eq:histogram_transition_component}
\begin{aligned}
H_{ij,\sigma'}^{r}
=
\sum\nolimits_{\sigma\in\mathcal S_{\mathrm c}}
\Big[
&a_{ij,\sigma}^{r}
\mathds{1}\{\delta(\sigma,1)=\sigma'\}\\
&+
\big(n_{\sigma}^{(i)}-a_{ij,\sigma}^{r}\big)
\mathds{1}\{\delta(\sigma,0)=\sigma'\}
\Big].
\end{aligned}
\end{equation}
and \(r=1,\ldots,|\mathcal A_{ij}|\) indexes the distinct input-count vectors.

\begin{defi}[Histogram edge]
\label{def:histogram_edge}
For each input-count vector
\(\bm a_{ij}^{r}\in\mathcal A_{ij}\), the directed edge \(e_{ij}^{r}\) is called a
\emph{histogram edge}, representing the multi-user state transition from vertex \(v_i\) to vertex \(v_j\) specified by \(\bm a_{ij}^{r}\). Consequently, all histogram edges are collected in $\mathcal E_{\mathrm H}=
    \left\{e_{ij}^{r}
    \,\middle|\,
    v_i,v_j\in\mathcal V_{\mathrm H},\;
    r=1,\ldots,|\mathcal A_{ij}|
    \right\}$.
\end{defi}


Each histogram edge retains the computation and signal information associated with the corresponding multi-user transition. Since \(a_{ij,\sigma}^{r}\) users in encoder state \(\sigma\) have input bit \(1\), the arithmetic-sum label carried by \(e_{ij}^{r}\) is \begin{equation} 
\label{eq:histogram_sum_label} S_{ij}^{r} = \sum\nolimits_{\sigma\in\mathcal S_{\mathrm c}} a_{ij,\sigma}^{r}. \end{equation}

Because the histogram representation removes user identities, the same histogram edge may correspond to multiple ordered multi-user input combinations. The number of such realizations is represented by the edge multiplicity 
\begin{equation} 
\label{eq:histogram_edge_multiplicity} \varphi_{ij}^{r} = \prod\nolimits_{\sigma\in\mathcal S_{\mathrm c}} \binom{n_\sigma^{(i)}}{a_{ij,\sigma}^{r}}. 
\end{equation} 

For a given rectangular ratio \(\rho\), the noiseless aggregated symbol associated with \(e_{ij}^{r}\) is presented as
\begin{equation} 
\label{eq:histogram_aggregate_symbol} 
\begin{aligned} 
u_{ij}^{r}(\rho) = \sum\nolimits_{\sigma\in\mathcal S_{\mathrm c}} \Big[ &\big(n_\sigma^{(i)}-a_{ij,\sigma}^{r}\big) \mu(\gamma(\sigma,0);\rho)\\ &+ a_{ij,\sigma}^{r} \mu(\gamma(\sigma,1);\rho) \Big]. 
\end{aligned} 
\end{equation} 

Accordingly, for the illustrative two-user example, the histogram-state transition graph contains three vertices
\(v_1,v_2,v_3\), associated with the histogram states
\begin{equation}
\label{eq:two_user_histogram_vertices}
    \bm n^{(i)}
    =
    \begin{cases}
        (2,0), & i=1,\\
        (1,1), & i=2,\\
        (0,2), & i=3.
    \end{cases}
\end{equation}

The resulting histogram-state transition
multigraph is illustrated in
Fig.~\ref{fig:histogram_based_state_transition}.

\input{Figures/trellis_diagram}

\section{Histogram-State Decoding}
\label{sec:histogram_state_decoding}

Based on the histogram-state representation developed in last section, we formulate the \ac{MAP} arithmetic-sum sequence criterion and derive the corresponding log-max \ac{MAP} histogram-state Viterbi decoder.

\subsection{Arithmetic-Sum Sequence Criterion}
\label{subsec:map_sum_sequence}

We first formulate the statistical decision criterion for recovering the
arithmetic-sum sequence from the received aggregated sequence $\mathbf y$.
The decoding objective is to recover a valid
arithmetic-sum sequence $\mathbf S$, we therefore
adopt a sequence-oriented \ac{MAP} criterion. Although this criterion does not
directly minimize the \ac{NMSE} of the reconstructed computation value, it provides
a tractable sequence-decoding objective, while the \ac{NMSE} evaluated in
Section~\ref{sec:simulation} is used to assess the resulting computation accuracy.

To account for the encoder memory, the histogram-state transition
multigraph \(\mathcal G_{\mathrm H}\) is unfolded over the \(T\)
encoder inputs defined in Section~II-A, forming a histogram-state
trellis of $T+1$ stages, where each stage contains $Q_{\mathrm H}$ histogram stages. For
\(\tau\in\{1,\ldots,T\}\), the encoder input \(b_k[\tau]\) is processed during the transition from stage \(\tau-1\) to stage \(\tau\).

Let \(\mathcal P_T\) denote the set of all terminated histogram paths
in this trellis. For a path \(\Pi\in\mathcal P_T\), let
\(q_\tau\in\{1,\ldots,Q_{\mathrm H}\}\) denote the histogram-state
index visited at stage \(\tau\). The
corresponding path transition is represented by
\(e_{q_{\tau-1}q_\tau}^{r_\tau}\), where \(r_\tau\) denotes the
selected parallel edge index. Hence, the terminated path can be written as $\Pi
    =(
    e_{q_0q_1}^{r_1},
    \ldots,
    e_{q_{T-1}q_T}^{r_T})$,
where \(q_0=q_T=1\) corresponds to the all-zero initial and terminal histogram state \(\bm n^{(1)}=(K,0,\ldots,0)\).
The first \(B\) edges form the information segment
    $\Pi_B
    =(
    e_{q_0q_1}^{r_1},
    \ldots,
    e_{q_{B-1}q_B}^{r_B})$,
while the last \(m\) edges correspond to encoder termination. 
The arithmetic-sum sequence generated by the information segment
\(\Pi_B\) is $\mathcal M(\Pi_B)
    =(
    S_{q_0q_1}^{r_1},
    \ldots,
    S_{q_{B-1}q_B}^{r_B})$.
Thus, a terminated histogram path \(\Pi\) produces a
candidate arithmetic-sum sequence \(\bm S\) if
\(\mathcal M(\Pi_B)=\bm S\), and we collect all paths in $\mathcal P(\bm S)=\left\{\Pi\in\mathcal P_T\,\middle|\,
    \mathcal M(\Pi_B)=\bm S\right\}$.

Since the histogram representation removes user identities, a
terminated histogram path can represent multiple ordered multi-user
input sequences. For a path $\Pi$, the
number of such realizations is
\begin{equation}
    \Omega(\Pi)
    =
    \prod\nolimits_{\tau=1}^{T}
    \varphi_{q_{\tau-1}q_\tau}^{r_\tau}.
\end{equation}



Under the assumption that the \(KB\) information bits are mutually
independent and equiprobable, 
the prior probability of each terminated histogram path can be evaluated as
\begin{equation}
    \Pr(\Pi)
    =
    2^{-KB}\Omega(\Pi).
    \label{eq:histogram_path_prior}
\end{equation}
Moreover, all ordered multi-user realizations represented by the same
histogram path generate the same noiseless aggregated symbol sequence. Thus, both the
path prior and the conditional likelihood can be evaluated from the histogram path without retaining the individual user identities.
Since the desired computation output is the arithmetic-sum sequence, the \ac{MAP} sequence decoder marginalizes the posterior probabilities of
all terminated histogram paths associated with the same \(\bm S\) with the following criterion
\begin{subequations}
  \begin{align}
    \widehat{\bm S}
    &=
    \arg\max_{\bm S}
    \sum\nolimits_{\Pi\in\mathcal P(\bm S)}
    \Pr(\Pi\mid\bm y)\\
    &=
    \arg\max_{\bm S}
    \sum\nolimits_{\Pi\in\mathcal P(\bm S)}
    \exp\bigl(-\Lambda(\Pi)\bigr)\\
    &=
    \arg\min_{\bm S}
    \left\{
    -\log
    \sum\nolimits_{\Pi\in\mathcal P(\bm S)}
    \exp\bigl(-\Lambda(\Pi)\bigr)
    \right\},
    \label{eq:map_sum_sequence_log_metric}
\end{align}  
\end{subequations}
where $\Lambda(\Pi)$ is defined as the additive metric of a terminated histogram path, which is given by
\begin{equation}\nonumber
    \Lambda(\Pi)
    =
    \sum\nolimits_{\tau=1}^{T}
    \left[
    \frac{
    \left|
    \vec y[\tau]
    -
    u_{q_{\tau-1}q_\tau}^{r_\tau}(\rho)
    \right|^2
    }{N_0}
    -
    \log
    \varphi_{q_{\tau-1}q_\tau}^{r_\tau}
    \right].
    \label{eq:histogram_path_metric}
\end{equation}

To obtain a tractable decoder,
we apply the log-max approximation, i.e.,
\(
-\log\sum_i\exp(-a_i)\approx\min_i a_i
\),
which gives
\begin{equation}
    \widehat{\bm S}\approx
    \arg\min_{\bm S}
    \min_{\Pi\in\mathcal P(\bm S)}
    \Lambda(\Pi)
    \label{eq:max_log_sum_sequence}
\end{equation}

Since the sets \(\{\mathcal P(\bm S)\}\) partition
\(\mathcal P_T\) according to $\bm S$, the
nested minimization in~\eqref{eq:max_log_sum_sequence} reduces to
\begin{equation}
    \widehat{\Pi}
    =
    \arg\min_{\Pi\in\mathcal P_T}
    \Lambda(\Pi),
    \label{eq:min_metric_histogram_path}
\end{equation}
with the decoded arithmetic-sum sequence given by 
$\widehat{\bm S}=\mathcal M(\widehat{\Pi}_B)$.
Under the log-max approximation, the original \ac{MAP} arithmetic-sum sequence
decoding is approximated by a minimum-metric path search on the
histogram-state trellis. 

\subsection{Log-Max MAP Histogram-State Decoding}
\label{subsec:histogram_viterbi}

The minimum-metric terminated histogram path problem in
\eqref{eq:min_metric_histogram_path} can be solved exactly by applying the
Viterbi algorithm to the histogram-state trellis~\cite{forney2005viterbi}. For
\(\tau\in \{1,\ldots,T\}\), let \(\mathcal E_{\mathrm H}^{(\tau)}\) denote
the set of histogram edges between stages \(\tau-1\) and
\(\tau\), defined as
\begin{equation}
    \mathcal E_{\mathrm H}^{(\tau)}
    =
    \begin{cases}
        \mathcal E_{\mathrm H},
        & 1\leq \tau\leq B,\\[1mm]
        \left\{
        e_{ij}^{r}\in\mathcal E_{\mathrm H}
        \,\middle|\,
        \bm a_{ij}^{r}=\bm 0
        \right\},
        & B<\tau\leq T.
    \end{cases}
    \label{eq:stage_histogram_edge_set}
\end{equation}

For \(e_{ij}^{r}\in\mathcal E_{\mathrm H}^{(\tau)}\), the branch
metric associated with the received symbol \(\vec y[\tau]\) is
\begin{equation}
    \operatorname{BM}_\tau(e_{ij}^{r})
    =
    \frac{
    \left|
    \vec y[\tau]-u_{ij}^{r}(\rho)
    \right|^2
    }{N_0}
    -
    \log \varphi_{ij}^{r}.
    \label{eq:viterbi_branch_metric}
\end{equation}
Here, \(i\) and \(j\) denote arbitrary source and target histogram-state
indices in the Viterbi recursion.

Let \(\operatorname{PM}_\tau(j)\) denote the minimum path metric reaching histogram
state \(\bm n^{(j)}\) at stage \(\tau\). Since the trellis starts from
the all-zero histogram state
\(\bm n^{(1)}=(K,0,\ldots,0)\), the path metrics are initialized as
\begin{equation}
    \operatorname{PM}_0(j)
    =
    \begin{cases}
        0, & j=1,\\
        +\infty, & \text{otherwise}.
    \end{cases}
    \label{eq:viterbi_initialization}
\end{equation}

For target histogram state \(\bm n^{(j)}\), define its corresponding
incoming edge set between stages \(\tau-1\) and \(\tau\) as $\mathcal I_\tau(j)
    =\{
    (i,r):
    e_{ij}^{r}\in\mathcal E_{\mathrm H}^{(\tau)}\}$.
The \ac{ACS} operation updates the path metric according to
\begin{equation}
    \operatorname{PM}_\tau(j)
    =
    \min_{(i,r)\in\mathcal I_\tau(j)}
    \left[
    \operatorname{PM}_{\tau-1}(i)
    +
    \operatorname{BM}_\tau(e_{ij}^{r})
    \right],
    \label{eq:viterbi_path_metric}
\end{equation}
and stores the corresponding survivor information as
\begin{equation}
    \operatorname{Sur}_{\tau}(j)
    =
    \arg\min_{(i,r)\in\mathcal I_\tau(j)}
    \left[
    \operatorname{PM}_{\tau-1}(i)
    +
    \operatorname{BM}_\tau(e_{ij}^{r})
    \right].
    \label{eq:viterbi_survivor}
\end{equation}
Thus, if
\(\operatorname{Sur}_{\tau}(j)=(i^\star,r^\star)\), the survivor
reaching \(\bm n^{(j)}\) at stage \(\tau\) comes from
\(\bm n^{(i^\star)}\) through histogram edge
\(e_{i^\star j}^{r^\star}\).
After stage \(T\), there is only one surviving path with the lowest path metric. Finally, the decoder traces back the minimum-metric
terminated histogram path \(\widehat{\Pi}\) in
\eqref{eq:min_metric_histogram_path}, and reconstructs the final computation value  according to
\eqref{eq:reconstructed_computation_sum}. Hence, the Viterbi recursion
exactly solves the minimum-metric path problem obtained from the
log-max approximation of the \ac{MAP} arithmetic-sum sequence criterion.

To illustrate the decoding procedure, we consider \(B=3\) input information bits for the two-user example in Fig.~\ref{fig:histogram_based_state_transition}. 
Fig.~\ref{fig:trellis diagram} illustrates a decoding instance for the noiseless
received sequence \(\bm y=[0,-2,0,2]\), where the Viterbi algorithm selects the terminated histogram path highlighted in the red line, and finally returns the recovered sum value
\(\widehat{x}_{\Sigma}=9\).

\begin{rem}
\label{rem:posterior means}
Note that the histogram-state trellis is not restricted to Viterbi decoding.
It can also support soft-output forward--backward inference, such as
the \ac{BCJR} algorithm~\cite{bahl1974optimal}, to obtain the posterior
marginals of the arithmetic-sum symbols $S[\tau]$. Since the quantized
computation value is linear in these
symbols, their posterior means can be used to construct the Bayesian
\ac{MMSE} estimate of the quantized sum~\cite{steven1993fundamentals}.
However, the resulting posterior-mean symbols generally do not form
a valid arithmetic-sum sequence associated with a terminated histogram
path. Moreover, in this work, we focus on sequence-oriented Viterbi decoding
because it provides a direct connection between
the computational error events and computational free-distance
criterion developed in Section~\ref{sec:computational_error_events}. Soft-output Bayesian decoding and
modulation design directly targeting computation \ac{MSE} are left for
future work.
\end{rem}

\subsection{Decoding Complexity}

We next analyze the computational complexity of the log-max \ac{MAP} histogram-state Viterbi decoder, and we first characterize the cardinality of the histogram edge set. 

\begin{prop}
\label{prop:histogram_edge_cardinality} 
For the \(K\)-user histogram-state transition multigraph generated by a convolutional encoder with \(N\) states, the total number of histogram edges is 
\begin{equation} 
C_{\mathrm H}=|\mathcal E_{\mathrm H}| = \binom{K+2N-1}{2N-1}. \label{eq:number_histogram_edges} 
\end{equation} 
\end{prop} 

\begin{proof}
    See Appendix~\ref{proof:histogram_edge_cardinality}
\end{proof}

At each trellis stage, the ACS operation evaluates the
incoming candidate histogram edges for all target states.
Since the total number of such edges is at most
\(C_{\mathrm H}\), decoding over the \(T\)
transitions requires at most
\(TC_{\mathrm H}\) branch evaluations.
Therefore, the computational complexity of the
histogram-state Viterbi decoder is
\begin{equation}
\mathcal O\!\left(
T
\binom{K+2N-1}{2N-1}
\right).
\label{eq:histogram_decoding_complexity}
\end{equation}
For comparison, the full-state joint convolutional decoder
in~\cite{you2023broadband} contains $N^K$ joint states, each with
$2^K$ outgoing transitions, resulting in a decoding complexity of
$\mathcal{O}\!\left(T(2N)^K\right)$. In contrast, for a fixed encoder
memory, the histogram-state decoder complexity in~\eqref{eq:histogram_decoding_complexity} grows polynomially with $K$, indicating that the histogram-state representation avoids the exponential growth with the number of users in both the decoder state space and the number of trellis transitions processed by the \ac{ACS} operation.


\section{Computational Free Distance Via Product
Graph Analysis}
\label{sec:computational_error_events}

In this section, we characterize the geometric distinguishability
of different computation outputs under ideal noiseless superposition.
We introduce the computational free distance and construct a product
graph for its evaluation. This distance will be used as the modulation
design objective in Section~\ref{sec:modulation_design}.

\subsection{Computational Error Events and Free Distance}

For standard convolutional codes, the free distance characterizes
the minimum Euclidean separation between two distinct coded paths that
depart from a common state and subsequently converge~\cite{viterbi1971convolutional}. In \ac{HiCoMAC}, however, distinct
histogram paths do not constitute a computational error if they produce
the same arithmetic-sum sequence. The relevant distance criterion
must therefore distinguish histogram paths according to their
computation outputs~\cite{zehavi1987performance}. We first distinguish histogram paths according to the arithmetic-sum
sequences carried by their information transitions.

\begin{defi}[Computational distinguishability]
\label{def:computational_distinguishability}
Two terminated histogram paths
\(\Pi,\widetilde{\Pi}\in\mathcal P_T\) are computationally equivalent
if $\mathcal M(\Pi_B)
    =
    \mathcal M(\widetilde{\Pi}_B)$,
and computationally distinct otherwise.
\end{defi}

Over a noiseless \ac{MAC}, a computational error can occur when
computationally distinct histogram paths generate identical aggregate
symbol sequences. To characterize and avoid such destructive
overlaps, we evaluate the separation between pairs of histogram paths. Since a pair of computationally distinct terminated histogram paths may diverge and converge multiple times, we partition the
path pair at every common histogram state, so that each resulting
segment pair starts from a common state and converges without an intermediate common state. Moreover, if the complete paths are computationally distinct, at
least one of the resulting segment pairs must produce different
arithmetic-sum label sequences. We
therefore focus on these irreducible segment pairs. Let $\pi=(e_{q_0q_1}^{r_1},\ldots,e_{q_{\nu-1}q_\nu}^{r_\nu})$ and $\widetilde{\pi}
=(e_{\widetilde q_0\widetilde q_1}^{\widetilde r_1},
\ldots,e_{\widetilde q_{\nu-1}\widetilde q_\nu}^{\widetilde r_\nu})$
denote two histogram path segments of length \(\nu\), with
arithmetic-sum sequences \(\mathcal M(\pi)\) and
\(\mathcal M(\widetilde{\pi})\), respectively. 

\begin{defi}[Irreducible computational error event]
\label{def:irreducible_computational_error_event}
The pair \((\pi,\widetilde{\pi})\) is an irreducible computational
error event if
\begin{equation}
    \left\{
    \begin{aligned}
        &q_0=\widetilde q_0,~
        q_\nu=\widetilde q_\nu,\\
        &q_h\neq\widetilde q_h,
        ~ h=1,\ldots,\nu-1,\\
        &\mathcal M(\pi)\neq\mathcal M(\widetilde{\pi}).
    \end{aligned}
    \right.
    \label{eq:irreducible_event_conditions}
\end{equation}
\end{defi}

Mathematically, we measure the separation of this event by the squared Euclidean distance between the two noiseless aggregated
symbols generated by the paired histogram transitions
\begin{equation}
    D(\pi,\widetilde{\pi};\rho)
    =
    \sum\nolimits_{h=1}^{\nu}
    D_{q_{h-1}q_h,\,
    \widetilde q_{h-1}\widetilde q_h}
    ^{r_h,\widetilde r_h}(\rho),
    \label{eq:computational_event_distance}
\end{equation}
where
\begin{equation}
    D_{ij,\widetilde i\widetilde j}^{r,\widetilde r}(\rho)
    =
    \left|
    u_{ij}^{r}(\rho)
    -
    u_{\widetilde i\widetilde j}^{\widetilde r}(\rho)
    \right|^2.
\label{eq:paired_histogram_edge_distance}
\end{equation}

In particular, \(D(\pi,\widetilde{\pi};\rho)=0\) indicates a destructive overlap of the two computationally distinct path segments  in the
noiseless aggregated signal space. To characterize the minimum separation over all such events, we define the following computational free distance.

\begin{defi}[Computational free distance]
\label{def:computational_free_distance}
Let \(\mathcal C_{\mathrm{comp}}\) denote the set of all finite
irreducible computational error events. For a given ratio \(\rho\),
the computational free distance is
\begin{equation}
    D_{\mathrm{comp}}(\rho)
    =
    \min_{(\pi,\widetilde{\pi})\in\mathcal C_{\mathrm{comp}}}
    D(\pi,\widetilde{\pi};\rho).
    \label{eq:computational_free_distance}
\end{equation}
\end{defi}

\begin{rem}
\label{rem:finite_block_distance}
Note that the minimization in
Definition~\ref{def:computational_free_distance} is over all
irreducible computational error events, without requiring the event
to start from the all-zero state, span the complete terminated path.
Consider any pair of computationally distinct terminated paths in
\(\mathcal P_T\). The path pair can be partitioned at its common
histogram states into segment pairs, at least one of which is an
irreducible computational error event. Since the squared Euclidean
distance between the terminated paths is the sum of the nonnegative
distances of these segment pairs, \(D_{\mathrm{comp}}(\rho)\)
lower bounds the minimum distance between computationally distinct
terminated paths and is used as the modulation design criterion.
\end{rem}

Evaluating \(D_{\mathrm{comp}}(\rho)\) requires searching over
irreducible computational error events departing from all possible
common histogram states. Since  the event distance in
\eqref{eq:computational_event_distance} is additive over paired
histogram transitions, we would show that these path pairs can instead be represented
as directed paths in a weighted product graph. In the next
subsections, we construct this graph and show that
\(D_{\mathrm{comp}}(\rho)\) equals its minimum accumulated path
weight between the divergence and convergence vertices.

\subsection{Product Graph Representation}
\label{subsec:product_graph_representation}

\subsubsection{Product Multigraph for Paired Histogram Transitions}
\label{subsubsec:product_multigraph}

This subsection constructs a directed product multigraph to represent
the path segment pairs in
\(\mathcal C_{\mathrm{comp}}\)~\cite{diallo2011efficient}.
Consider the histogram-state transition multigraph
\(\mathcal G_{\mathrm H}
=(\mathcal V_{\mathrm H},\mathcal E_{\mathrm H})\).
The product graph associated with pairs of histogram path segments is
a directed multigraph
\(\mathcal G_{\mathrm P}
=(\mathcal V_{\mathrm P},\mathcal E_{\mathrm P})\).
For any pair of path segments, the product vertex
\(p_{i\widetilde i}^{\eta}\) records the histogram state pairs \(\boldsymbol n^{(i)}\) and
\(\boldsymbol n^{(\widetilde i)}\), together with a binary flag
\(\eta\). The flag \(\eta=0\) indicates that their arithmetic-sum
label sequences are identical so far, whereas \(\eta=1\) indicates
that they have differed in at least one transition. Accordingly, we define the product vertex set as
\begin{equation}
    \mathcal V_{\mathrm P}
    =
    \left\{
    p_{i\widetilde i}^{\eta}
    \,\middle|\,
    i,\widetilde i\in\{1,\ldots,Q_{\mathrm H}\},
    \ \eta\in\{0,1\}
    \right\},
    \label{eq:product_vertex_set}
\end{equation}
where we distinguish product
vertices according to whether the two histogram state are
identical. A product vertex $p_{ii}^{\eta}$ is referred to as a \emph{diagonal
product vertex}, while $p_{i\tilde i}^{\eta}$
with $i \neq \tilde i$ is referred to as an \emph{off-diagonal product vertex}.

For a product vertex \(p_{i\widetilde i}^{\eta}\), consider a pair
of histogram edges \(e_{ij}^{r}\) and
\(e_{\widetilde i\widetilde j}^{\widetilde r}\), taken by the first
and second path segments, respectively. This paired transition gives
a directed product edge from \(p_{i\widetilde i}^{\eta}\) to
\(p_{j\widetilde j}^{\eta'}\), where
\begin{equation}
    \eta'
    =
    \max
    \left\{
    \eta,\,
    \mathbbm 1
    \left[
    S_{ij}^{r}
    \neq
    S_{\widetilde i\widetilde j}^{\widetilde r}
    \right]
    \right\}.
    \label{eq:product_flag_update}
\end{equation}
The corresponding product edge is denoted by
\(\zeta_{i\widetilde i,j\widetilde j}
^{\eta,\eta',r,\widetilde r}\), and collecting all pairs of distinct histogram edges gives
\begin{equation}
\begin{aligned}
    \mathcal E_{\mathrm P}
    =
    \Bigl\{
    &\zeta_{i\widetilde i,j\widetilde j}
    ^{\eta,\eta',r,\widetilde r}
    \,\Bigm|\,
    e_{ij}^{r},
    e_{\widetilde i\widetilde j}^{\widetilde r}
    \in\mathcal E_{\mathrm H}, ~e_{ij}^{r} \neq
    e_{\widetilde i\widetilde j}^{\widetilde r}
    \Bigr\}.
    \label{eq:product_edge_set}
\end{aligned}
\end{equation}

Correspondingly, for a product edge
\(\zeta_{i\widetilde i,j\widetilde j}
^{\eta,\eta',r,\widetilde r}\),
its squared Euclidean distance can be calculated as $D_{ij,\widetilde i\widetilde j}^{r,\widetilde r}(\rho)$ based on~\eqref{eq:paired_histogram_edge_distance}.
Accordingly, the accumulated
distance along the product graph path corresponding to
\((\pi,\widetilde{\pi})\) therefore equals
\(D(\pi,\widetilde{\pi};\rho)\) in
\eqref{eq:computational_event_distance}.

Each irreducible computational error event in
\(\mathcal C_{\mathrm{comp}}\) determines a directed path in
\(\mathcal G_{\mathrm P}\). However, the full product multigraph also
contains paired paths that do not satisfy the irreducibility
conditions in
Definition~\ref{def:irreducible_computational_error_event}. We therefore construct a modified product graph by containing only divergence, off-diagonal, and convergence
transitions.

\subsubsection{Modified Product Graph for Irreducible Events}
\label{subsubsec:modified_product_graph}

We first identify the diagonal product vertices to be merged. Before
divergence, both paths occupy the same histogram state and their
arithmetic-sum label sequences are identical, corresponding to a
diagonal product vertex \(p_{ii}^{0}\). Such a vertex can initiate an irreducible event if a pair of distinct outgoing histogram edges either reaches different histogram states or reaches the same state with different arithmetic-sum labels. Accordingly, the corresponding divergence vertices are collected in
\begin{align}
    \mathcal V_{\mathrm{div}}
    ={}&
    \left\{
    p_{ii}^{0}
    \,\middle|\,
    \exists\,
    \zeta_{ii,jj}^{0,1,r,\widetilde r}
    \in\mathcal E_{\mathrm P}
    \right\}
    \nonumber\\
    &{}\cup
    \left\{
    p_{ii}^{0}
    \,\middle|\,
    \exists\,
    \zeta_{ii,j\widetilde j}
    ^{0,\eta',r,\widetilde r}
    \in\mathcal E_{\mathrm P},
    \ j\neq\widetilde j
    \right\}.
    \label{eq:divergence_vertex_set}
\end{align}

An irreducible event terminates when the path segments reach the same
histogram state after their arithmetic-sum label sequences have
differed. This convergence may occur after the paths have occupied different
histogram states or 
directly in the divergence transition. Accordingly, the corresponding convergence vertices are collected in
\begin{align}
    \mathcal V_{\mathrm{conv}}
    ={}&
    \left\{
    p_{jj}^{1}
    \,\middle|\,
    \exists\,
    \zeta_{i\widetilde i,jj}^{\eta,1,r,\widetilde r}
    \in\mathcal E_{\mathrm P},
    \ i\neq\widetilde i
    \right\}
    \nonumber\\
    &{}\cup
    \left\{
    p_{jj}^{1}
    \,\middle|\,
    \exists\,
    \zeta_{ii,jj}^{0,1,r,\widetilde r}
    \in\mathcal E_{\mathrm P},
    \ p_{ii}^{0}\in\mathcal V_{\mathrm{div}}
    \right\}.
    \label{eq:convergence_vertex_set}
\end{align}

Since the computational free distance considers irreducible events departing from any common histogram state and terminating at any first valid convergence, we merge all vertices in \(\mathcal V_{\mathrm{div}}\) into a single initial vertex \(p_{\mathrm{in}}\), and all vertices in \(\mathcal V_{\mathrm{conv}}\) into a single final vertex \(p_{\mathrm{out}}\). The off-diagonal product vertices are retained as internal vertices and collected in
\begin{equation}
    \mathcal V_{\mathrm{off}}
    =
    \left\{
    p_{ij}^{\eta}
    \in\mathcal V_{\mathrm P}
    \,\middle|\,
    i\neq j
    \right\}.
    \label{eq:off_diagonal_product_vertex_set}
\end{equation}
These merged and retained vertices form the modified product vertex set $\mathcal V_{\mathrm M}=
    \mathcal V_{\mathrm{off}}
    \cup\{p_{\mathrm{in}},p_{\mathrm{out}}\}$.
After the vertex merging, a directed connection between two modified product vertices is referred to as a modified product edge. The possible modified product edges are
\begin{subequations}
\label{eq:modified_product_edge_types}
\begin{align}
    \varepsilon_{\mathrm{in},j\tilde j}^{\eta'}
    &=
    \left(
    p_{\mathrm{in}},
    p_{j\tilde j}^{\eta'}
    \right),
    ~j\neq\tilde j,
    \label{eq:initial_modified_product_edge}
    \\
    \varepsilon_{i\tilde i,j\tilde j}^{\eta,\eta'}
    &=
    \left(
    p_{i\tilde i}^{\eta},
    p_{j\tilde j}^{\eta'}
    \right),
    ~i\neq\tilde i, ~ j\neq\tilde j,
    \label{eq:internal_modified_product_edge}
    \\
    \varepsilon_{i\tilde i,\mathrm{out}}^{\eta}
    &=
    \left(
    p_{i\tilde i}^{\eta},
    p_{\mathrm{out}}
    \right),
    ~ i\neq\tilde i,
    \label{eq:terminal_modified_product_edge}
    \\
    \varepsilon_{\mathrm{in},\mathrm{out}}
    &=
    \left(
    p_{\mathrm{in}},
    p_{\mathrm{out}}
    \right).
    \label{eq:direct_modified_product_edge}
\end{align}
\end{subequations}

After the vertex merging, all product edges with the same resulting
source and destination vertices are merged into one modified product
edge, including parallel edges already present in
\(\mathcal G_{\mathrm P}\). For each of the four edge forms in
\eqref{eq:modified_product_edge_types}, the merged product edges construct
the corresponding candidate set, defined as
\begin{subequations}
\label{eq:modified_product_edge_candidate_sets}
\begin{align}
    \mathcal C_{\mathrm{in},j\widetilde j}^{\eta'}
    &=
    \left\{
    \zeta_{ii,j\widetilde j}^{0,\eta',r,\widetilde r}
    \in\mathcal E_{\mathrm P}
    \,\middle|\,
    p_{ii}^{0}\in\mathcal V_{\mathrm{div}}
    \right\},
    \label{eq:initial_candidate_set}\\
    \mathcal C_{i\widetilde i,j\widetilde j}
    ^{\eta,\eta'}
    &=
    \left\{
    \zeta_{i\widetilde i,j\widetilde j}
    ^{\eta,\eta',r,\widetilde r}
    \in\mathcal E_{\mathrm P}
    \right\},
    \label{eq:internal_candidate_set}\\
    \mathcal C_{i\widetilde i,\mathrm{out}}^{\eta}
    &=
    \left\{
    \zeta_{i\widetilde i,jj}^{\eta,1,r,\widetilde r}
    \in\mathcal E_{\mathrm P}
    \,\middle|\,
    p_{jj}^{1}\in\mathcal V_{\mathrm{conv}}
    \right\},
    \label{eq:terminal_candidate_set}\\
    \mathcal C_{\mathrm{in},\mathrm{out}}
    &=
    \left\{
    \zeta_{ii,jj}^{0,1,r,\widetilde r}
    \in\mathcal E_{\mathrm P}
    \,\middle|\,
    p_{ii}^{0}\in\mathcal V_{\mathrm{div}},\
    p_{jj}^{1}\in\mathcal V_{\mathrm{conv}}
    \right\}.
    \label{eq:direct_candidate_set}
\end{align}
\end{subequations}

The modified product edge set \(\mathcal E_{\mathrm M}\) contains all edges in \eqref{eq:modified_product_edge_types} with nonempty
candidate sets.
Therefore, we obtain a modified product graph 
$\mathcal G_{\mathrm M}=\left(\mathcal V_{\mathrm M},\mathcal E_{\mathrm M}\right)$, which contains all transitions needed for evaluating the computational free distance $D_{\rm comp}$. 

\subsection{Fixed Ratio Shortest Path Characterization}
\label{subsec:shortest_path_characterization}


Consider a modified product edge \(\varepsilon\in\mathcal E_{\mathrm M}\)
of any form and let \(\mathcal C_{\varepsilon}\) denote its
corresponding candidate set, with the indices omitted for brevity.
Since all product edges $\zeta$ in \(\mathcal C_{\varepsilon}\) have the same
modified source and destination vertices, the weight assigned to
\(\varepsilon\) is defined as the minimum of their distances:
\begin{equation}
    w_{\mathrm M}(\varepsilon;\rho)
    =
    \min_{\zeta\in\mathcal C_{\varepsilon}}
    D_{\zeta}(\rho).
    \label{eq:modified_product_edge_weight}
\end{equation}

With the weight $w_{\mathrm M}(\varepsilon;\rho)$ assigned to the edge
\(\varepsilon\), we obtain the weighted modified product graph $\mathcal G_{\mathrm M}(\rho)=
    \left(
    \mathcal V_{\mathrm M},
    \mathcal E_{\mathrm M},
    w_{\mathrm M}(\cdot;\rho)
    \right)$.
Let \(\mathcal P_{\mathrm M}\) denote the set of all directed paths
from \(p_{\mathrm{in}}\) to \(p_{\mathrm{out}}\) in
\(\mathcal G_{\mathrm M}(\rho)\). By construction of the modified product graph, every irreducible
computational error event
\((\pi,\widetilde{\pi})\in\mathcal C_{\mathrm{comp}}\)
determines a directed path $\omega
    =\left(
    \varepsilon^{(1)},\ldots,\varepsilon^{(\nu)}
    \right)
    \in\mathcal P_{\mathrm M}$.
Conversely, selecting one product edge from the candidate set of each
modified product edge along \(\omega\) determines a corresponding
pair of histogram path segments. For the fixed ratio \(\rho\), the accumulated distance along
\(\omega\) is defined as
\begin{equation}
    W_{\mathrm M}(\omega;\rho)
    =
    \sum\nolimits_{h=1}^{\nu}
    w_{\mathrm M}
    \left(
    \varepsilon^{(h)};\rho
    \right).
    \label{eq:modified_product_path_weight}
\end{equation}

Since each modified product edge weight is the minimum distance over its candidate set, \(W_{\mathrm M}(\omega;\rho)\) is the minimum
distance among the irreducible computational error events represented by \(\omega\). Hence, the computational free distance is
obtained by minimizing \(W_{\mathrm M}(\omega;\rho)\) over \(\mathcal P_{\mathrm M}\).

\input{Figures/modified_product_graph}

\begin{theorem}[shortest path characterization of computational
free distance]
\label{thm:shortest_path_characterization}
For any fixed rectangular ratio \(\rho\), the computational free
distance satisfies
\begin{equation}
    D_{\mathrm{comp}}(\rho)
    =
    \min_{\omega\in\mathcal P_{\mathrm M}}
    W_{\mathrm M}(\omega;\rho).
    \label{eq:shortest_path_characterization}
\end{equation}
Hence, \(D_{\mathrm{comp}}(\rho)\) is the shortest path distance from
\(p_{\mathrm{in}}\) to \(p_{\mathrm{out}}\) in
\(\mathcal G_{\mathrm M}(\rho)\).
\end{theorem}

\begin{proof}
See Appendix~\ref{app:shortest_path_characterization}.
\end{proof}

For the illustrative two-user example, 
Fig.~\ref{fig:modified_product_graph} provides a graphical
representation of the weighted modified product graph constructed
from the transition graph in
Fig.~\ref{fig:histogram_based_state_transition}.

\section{Modulation Design and Solution}
\label{sec:modulation_design}

This section develops the computation-oriented modulation design of
\ac{HiCoMAC} based on the computational free distance characterized in
Section~\ref{sec:computational_error_events}. We first formulate the
rectangular-ratio design problem, then develop a coarse-to-fine search
based on the shortest path characterization, followed by an analysis of
its offline computational complexity.

\subsection{Modulation Design Formulation}
\label{subsec:ratio_parameterized_distances}

In this subsection, we design the scaling of the \ac{QAM} constellation by optimizing its rectangular ratio. As discussed in Remark~\ref{rem:finite_block_distance},
\(D_{\mathrm{comp}}(\rho)\) provides a lower bound on the minimum computational distance of the finite terminated trellis. To enhance this computational distinguishability under a fixed average symbol energy, we pose the constellation design problem as follows\footnote{The present geometric criterion may be further
extended by incorporating the multiplicities and prior
probabilities of computational error events, as well as the
different numerical significance of errors across bit positions,
into the modulation design.}
\begin{equation}
    \mathcal{P}_0:=\max_{\rho_{\min}\leq\rho\leq\rho_{\max}}
    D_{\mathrm{comp}}(\rho),
    \label{eq:rectangular_ratio_design}
\end{equation}
where \([\rho_{\min},\rho_{\max}]\) is a prescribed search interval.

For notational brevity, let \(e\) denote an arbitrary histogram edge,
and the ratio-independent coefficients 
are defined as
\begin{subequations}
\label{eq:histogram_edge_coefficients}
\begin{align}
    A_e
    &=
    \sum_{\sigma\in\mathcal S_{\mathrm c}}
    \Big[(n_\sigma^{(i)}-a_{ij,\sigma}^r)
    \widetilde{x}_{\lambda(\gamma(\sigma,0)),\mathrm I}
    +a_{ij,\sigma}^r \widetilde{x}_{\lambda(\gamma(\sigma,1)),\mathrm I}
    \Big], \nonumber \\ \nonumber
    B_e
    &=
    \sum_{\sigma\in\mathcal S_{\mathrm c}}
    \Big[(n_\sigma^{(i)}-a_{ij,\sigma}^r)
    \widetilde{x}_{\lambda(\gamma(\sigma,0)),\mathrm Q}
    +a_{ij,\sigma}^r \widetilde{x}_{\lambda(\gamma(\sigma,1)),\mathrm Q}
    \Big].
    \label{eq:histogram_edge_quadrature_coefficient}
\end{align}
\end{subequations}

Consider a product edge \(\zeta\) associated with the histogram edge
pair \((e,\widetilde e)\), the corresponding distance is given by
\begin{equation}
    D_{\zeta}(\rho)
    = \xi^2(\rho)
    (\rho^2\alpha_{\zeta}
    +
    \beta_{\zeta}),
    \label{eq:parametric_product_edge_distance}
\end{equation}
where $\alpha_{\zeta}
    =\left(A_e-A_{\widetilde e}\right)^2$ and $\beta_{\zeta}
    =\left(B_e-B_{\widetilde e}\right)^2$.
Accordingly, Problem
$\mathcal{P}_0$ can be written explicitly as
\begin{equation}
\begin{aligned}
    \nonumber
    \mathcal{P}_0:=\max_{\rho_{\min}\leq\rho\leq\rho_{\max}}
    ~\xi^2(\rho)
    \min_{\omega\in\mathcal P_{\mathrm M}}
    \sum_{h=1}^{\nu}
    \min_{\zeta\in\mathcal C_{\varepsilon^{(h)}}}
    \left(
        \rho^2\alpha_{\zeta}
        +\beta_{\zeta}
    \right).
\end{aligned}
\label{eq:parametric_rectangular_ratio_design}
\end{equation}

The component energies $E_I$ and $E_Q$
in $\xi(\rho)$ are evaluated once for the considered encoder, quantization precision and termination rule, and remain fixed throughout the ratio search.
For any trial ratio \(\rho\), the objective is evaluated by updating the modified product edge weights and computing the shortest path from \(p_{\mathrm{in}}\) to \(p_{\mathrm{out}}\) in
\(\mathcal G_{\mathrm M}(\rho)\). Since the objective can be nondifferentiable, We therefore employ a search based on an initial coarse grid followed by local grid refinement.

\subsection{Coarse-to-Fine Ratio Search}
\label{subsec:coarse_to_fine_ratio_search}

\begin{algorithm}[!t]
\caption{Coarse-to-Fine Ratio Search}
\label{alg:ratio_search}
\begin{algorithmic}[1]
\STATE \textbf{Input:}
Modified product graph
\(\mathcal G_{\mathrm M}\),
candidate sets
\(\{\mathcal C_{\varepsilon}\}_{\varepsilon\in\mathcal E_{\mathrm M}}\),
search interval
\([\rho_{\min},\rho_{\max}]\), tolerance
\(\epsilon_\rho\).

\STATE Precompute the product edge coefficients $\alpha_{\zeta}$ and $\beta_{\zeta}$.
\STATE Generate initial grid \(\mathcal R_0\). 
Set
\(\mathcal R_{\mathrm{eval}}\leftarrow\emptyset\)
and \(t\leftarrow 0\).

\FOR{each \(\rho\in\mathcal R_0\)}
    \STATE Evaluate \(D_{\mathrm{comp}}(\rho)\) using Dijkstra on \(\mathcal G_{\mathrm M}(\rho)\).
    \STATE
    Set \(\mathcal R_{\mathrm{eval}}
    \leftarrow
    \mathcal R_{\mathrm{eval}}\cup\{\rho\}\).
\ENDFOR

\STATE Select \(\hat{\rho}_0\) using
\eqref{eq:best_coarse_ratio}.

\WHILE{\(\Delta_t>\epsilon_\rho\)}
    \STATE Generate the local grid \(\mathcal R_{t+1}\) using
    \eqref{eq:local_refinement_interval}--\eqref{eq:local_grid_spacing}.

    \FOR{each
    \(\rho\in
    \mathcal R_{t+1}\setminus\mathcal R_{\mathrm{eval}}\)}
        \STATE Evaluate \(D_{\mathrm{comp}}(\rho)\) using Dijkstra on \(\mathcal G_{\mathrm M}(\rho)\).
        \STATE
        Set \(\mathcal R_{\mathrm{eval}}
        \leftarrow
        \mathcal R_{\mathrm{eval}}\cup\{\rho\}\).
    \ENDFOR

    \STATE Select \(\hat{\rho}_{t+1}\) using
    \eqref{eq:best_ratio_update}, and update $\Delta t$ using \eqref{eq:local_grid_spacing}.
    \STATE
    Set \(t\leftarrow t+1\).
\ENDWHILE

\STATE Set \(\hat{\rho}\leftarrow\hat{\rho}_t\).
\STATE \textbf{Output:}
\(\hat{\rho}\) and \(D_{\mathrm{comp}}(\hat{\rho})\).
\end{algorithmic}
\end{algorithm}

Let \(G\) denote the number of grid points used in each search round, and the initial grid is defined as
\begin{equation}
    \mathcal R_0
    =
    \left\{
    \rho_{\min}+i\Delta_0
    \Bigm|
    i=0,\ldots,G-1
    \right\},
    \label{eq:initial_ratio_grid}
\end{equation}
where
\begin{equation}
    \Delta_0
    =
    \frac{\rho_{\max}-\rho_{\min}}{G-1}.
    \label{eq:initial_grid_spacing}
\end{equation}
For each \(\rho\in\mathcal R_0\),
\(D_{\mathrm{comp}}(\rho)\) is evaluated using Dijkstra
algorithm~\cite{magzhan2013review} according to
Theorem~\ref{thm:shortest_path_characterization}.
After evaluating the initial grid, we set
\(\mathcal R_{\mathrm{eval}}=\mathcal R_0\), and select
\begin{equation}
\label{eq:best_coarse_ratio}
\hat{\rho}_0
=
    \arg\max_{\rho\in\mathcal R_{\mathrm{eval}}}
    D_{\mathrm{comp}}(\rho).
\end{equation}

At refinement round \(t\), let \(\hat{\rho}_t\) and
\(\Delta_t\) denote the currently selected ratio and grid spacing,
respectively. The next search interval is centered at
\(\hat{\rho}_t\) and restricted to the prescribed range:
\begin{subequations}
\label{eq:local_refinement_interval}
\begin{align}
    \rho_{\min}^{(t+1)}
    &=
    \max\left\{
    \rho_{\min},
    \hat{\rho}_t-\Delta_t
    \right\},
    \\
    \rho_{\max}^{(t+1)}
    &=
    \min\left\{
    \rho_{\max},
    \hat{\rho}_t+\Delta_t
    \right\}.
\end{align}
\end{subequations}

The local grid at refinement round \(t+1\) is defined as
\begin{equation}
    \mathcal R_{t+1}
    =
    \left\{
    \rho_{\min}^{(t+1)}+i\Delta_{t+1}
    \Bigm|
    i=0,\ldots,G-1
    \right\},
    \label{eq:local_ratio_grid}
\end{equation}
where
\begin{equation}
    \Delta_{t+1}
    =
    \frac{
    \rho_{\max}^{(t+1)}
    -
    \rho_{\min}^{(t+1)}
    }{
    G-1
    }.
    \label{eq:local_grid_spacing}
\end{equation}

Only the new points in
\(\mathcal R_{t+1}\setminus\mathcal R_{\mathrm{eval}}\) are evaluated.
The evaluated set and the selected ratio are then updated as
\begin{subequations}
\label{eq:ratio_search_update}
\begin{align}
    \mathcal R_{\mathrm{eval}}
    &\leftarrow
    \mathcal R_{\mathrm{eval}}\cup\mathcal R_{t+1},
    \\
    \hat{\rho}_{t+1}
    &=
    \arg\max_{\rho\in\mathcal R_{\mathrm{eval}}}
    D_{\mathrm{comp}}(\rho).
    \label{eq:best_ratio_update}
\end{align}
\end{subequations}

The refinement terminates when
\(\Delta_{t+1}\leq\epsilon_\rho\), where
\(\epsilon_\rho\) is the grid spacing tolerance. The final selected
ratio is denoted by \(\hat{\rho}\) and returned as an approximate
solution to $\mathcal{P}_0$. The complete procedure
is summarized in Algorithm~\ref{alg:ratio_search}. This coarse-to-fine
procedure provides an efficient approximate search over the prescribed
ratio interval, with the resolution controlled by
\(\epsilon_\rho\).

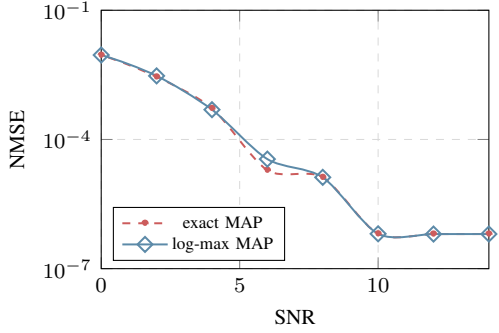
\begin{figure}[!t]
\centering
\begin{tikzpicture}
    \begin{axis}[
        xlabel = {SNR},
        ylabel = {NMSE},
        label style={font=\footnotesize},
        width=0.37\textwidth,
        height=5cm,
        xmin=0, xmax=14,
        ymin=0.0000001, ymax=0.1,
        legend style={nodes={scale=0.65, transform shape}, at={(0.3,0.85)}},
        ticklabel style = {font=\footnotesize},
        legend pos=south west,
        ymajorgrids=true,
        xmajorgrids=true,
        grid style=dashed,
        grid=both,
        ymode = log,
        grid style={line width=.1pt, draw=gray!10},
        major grid style={line width=.2pt,draw=gray!30},
    ]
    \addplot[ smooth,
             thin,
             dashed,
        color=chestnut,
        mark=*,
        line width=0.75pt,
        mark size=1pt,
        ]
    table[x=SNR,y=exact_map_nmse]
    {Data/MAP_accuracy.dat};
    \addlegendentry{exact MAP}
    \addplot[ smooth,
             thin,
             color=airforceblue,
        mark=square,
        mark options = {rotate = 45, solid},
        line width=0.75pt,
        mark size=2pt,
        ]
    table[x=SNR,y=maxlog_histogram_nmse]
    {Data/MAP_accuracy.dat};
    \addlegendentry{log-max MAP}
\end{axis}
\end{tikzpicture}
\vspace{-1em}
\caption{{\color{black}NMSE comparison between exact MAP decoding and log-max MAP
histogram-state decoding. 
}}
\label{fig:MAP accuracy}
\end{figure}

\vspace{-1em}

\subsection{Computational Complexity}

\begin{table*}[t]
    \centering
    \caption{Computational free distance and NMSE gain for different encoder and modulation configurations
    .}
    \label{tab:dcomp_nmse_gain_k2}
    \setlength{\tabcolsep}{4.2pt}
    \renewcommand{\arraystretch}{1.1}
    \begin{tabular}{|c|c|c|c|c|c|c|c|}
        \hline
        Memory order & Generators & Coding rate & QAM & $\hat{\rho}$ & $D_{\mathrm{comp}}(1)$ & $D_{\mathrm{comp}}(\hat{\rho})$ & $G_{\mathrm{NMSE}}$ at 3/5/7 dB \\
        \hline
        $1$ & $(1,3)_{8}$ & $1/2$ & 4QAM & 1.225 & 6.000 & 6.400 & 0.21/1.77/0.82 \\
        $3$ & $(1,15)_{8}$ & $1/2$ & 4QAM & 1.000 & 8.000 & 8.000 & 0.00/0.00/0.00 \\
        $2$ & $(1,3,5,7)_{8}$ & $1/4$ & 16QAM & 0.577 & 5.478 & 6.330 & 0.66/3.73/3.35 \\
        $3$ & $(1,11,13,17)_{8}$ & $1/4$ & 16QAM & 0.775 & 7.200 & 8.000 & 0.10/0.74/4.04 \\
        \hline
    \end{tabular}
    \vspace{1mm}
\end{table*}

We analyze the computational complexity of the
proposed modulation design. 
Since a product vertex records an ordered pair of histogram
states and a binary flag, the full product multigraph contains $|\mathcal V_{\mathrm P}|=2Q_{\mathrm H}^{2}$ vertices. Similarly, pairing two distinct histogram edges for each source flag gives $|\mathcal E_{\mathrm P}|=2C_{\mathrm H}(C_{\mathrm H}-1)$ edges.
The modified product graph retains the two flag values for each ordered pair of distinct histogram states, together with \(p_{\mathrm{in}}\) and \(p_{\mathrm{out}}\), and hence
$|\mathcal V_{\mathrm M}|=2Q_{\mathrm H}(Q_{\mathrm H}-1)+2$. Let \(Q_{\mathrm{off}}=Q_{\mathrm H}(Q_{\mathrm H}-1)\), and the total number of modified edges is bounded by
\begin{equation}
    |\mathcal E_{\mathrm M}|
    \leq
    \min\!\left\{
    |\mathcal E_{\mathrm P}|,
    3Q_{\mathrm{off}}^{2}+4Q_{\mathrm{off}}+1
    \right\}.
\end{equation}

The modified graph and its candidate sets are constructed by
a single traversal of the product edges. Parallel product edges
mapped to the same modified edge are represented by their
precomputed coefficient pairs
\((\alpha_\zeta,\beta_\zeta)\), resulting in
\(\mathcal O(|\mathcal E_{\mathrm P}|)\) construction
complexity and
\(\mathcal O(|\mathcal V_{\mathrm M}|+
|\mathcal E_{\mathrm M}|+
|\mathcal E_{\mathrm P}|)\) storage.
For each trial ratio \(\rho\), the modified edge weight update
requires at most \(\mathcal O(|\mathcal E_{\mathrm P}|)\)
operations, followed by a Dijkstra search with complexity
\(\mathcal O((|\mathcal V_{\mathrm M}|+
|\mathcal E_{\mathrm M}|)\log|\mathcal V_{\mathrm M}|)\).
Therefore, if \(N_{\mathrm{eval}}\) distinct ratios are
evaluated, the overall search complexity is
\begin{equation}
\label{eq:ratio_search_complexity}
\mathcal O\!\left(
N_{\mathrm{eval}}
\left[
|\mathcal E_{\mathrm P}|
+
(|\mathcal V_{\mathrm M}|+|\mathcal E_{\mathrm M}|)
\log|\mathcal V_{\mathrm M}|
\right]
\right).
\end{equation}
Above all, for a fixed convolutional encoder, the computational complexity of
Algorithm~\ref{alg:ratio_search} grows polynomially with \(K\). However, the product edge enumeration can
still become computationally demanding for large \(K\) or $m$. The proposed graph construction and ratio search are therefore intended as offline design for moderate problem sizes.

\section{Simulation Results}
\label{sec:simulation}

In this section, we evaluate the computational accuracy and decoding
complexity of \ac{HiCoMAC} through Monte Carlo simulations. The
computational accuracy by a fixed $\rho$ is measured by the \ac{NMSE}, defined as
\begin{equation}
\mathrm{\ac{NMSE}_{\rho}}
:=
\frac{
\sum_{j=1}^{N_s}
\left|x_{\Sigma}^{(j)}-\widehat{x}_{\Sigma}^{(j)}\right|^2
}{
N_s\left|x_{\Sigma}^{\max}-x_{\Sigma}^{\min}\right|^2
},
\label{eq:nmse}
\end{equation}
where $x_{\Sigma}^{(j)}$
is the desired function value, $\widehat{x}_{\Sigma}^{(j)}$ is the estimated value in the $j$-th Monte Carlo trial. \(x_{\Sigma}^{\max}\) and \(x_{\Sigma}^{\min}\) are the maximum
and minimum function outputs, respectively, and \(N_s\) is the number
of simulated input blocks. Each user input is independently drawn
from \([0,1]\) and quantized using \(B\) bits. The transmitted sequence is normalized according
to~\eqref{eq:normalized_rectangular_qam}, such that the per-user
average symbol energy over the \(T\) transmissions is \(E_{\mathrm s}=1\)
for every \(\rho\). Accordingly, the per-user \ac{SNR} is defined as
\(\mathrm{SNR}=E_{\mathrm s}/N_0\). The rectangular ratio is searched
over \(\rho\in[0.3,3.0]\). For each SNR point, the simulation terminates
when either \(N_{\min}^{\mathrm{err}}=100\) computational errors have been observed or
\(N_s=10^5\).

\vspace{-1em}

\subsection{Validation of the Log-Max Approximation}

We first validate the log-max approximation using a configuration with $K=2$ users, $B=8$ quantization bits, and a convolutional
encoder $g=(1,3)_8$ with memory order $m=1$, for which the exact MAP criterion can
be evaluated by exhaustively enumerating all terminated histogram paths. As shown in
Fig.~\ref{fig:MAP accuracy}, the log-max \ac{MAP} histogram-state decoder
closely follows the exact MAP decoder over the considered \ac{SNR} range.
The \ac{NMSE} difference becomes negligible
as the \ac{SNR} increases, indicating that the log-max approximation decoding provides a close
approximation to the exact MAP decoding performance for high \ac{SNR} condition.

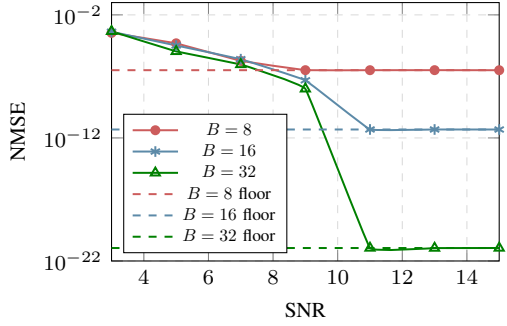
\begin{figure}[!t]
\centering
\begin{tikzpicture}
    \begin{axis}[
        xlabel = {SNR},
        ylabel = {NMSE},
        label style={font=\footnotesize},
        width=0.37\textwidth,
        height=5cm,
        xmin=3, xmax=15,
        ymin=0.0000000000000000000001, ymax=0.1,
        legend style={nodes={scale=0.65, transform shape}, at={(0.3,0.85)}},
        ticklabel style = {font=\footnotesize},
        legend pos=south west,
        ymajorgrids=true,
        xmajorgrids=true,
        grid style=dashed,
        grid=both,
        ymode = log,
        grid style={line width=.1pt, draw=gray!10},
        major grid style={line width=.2pt,draw=gray!30},
    ]
    \addplot[ smooth,
              tension=0.1,
             thin,
        color=chestnut,
        mark=*,
        line width=0.75pt,
        mark size=1.5pt,
        ]
    table[x=SNR,y=B8]
    {Data/trellis_b_sweep_4qam_nmse.dat};
    \addplot[ smooth,
              tension=0.1,
             thin,
        color=airforceblue,
        mark=asterisk,
        line width=0.75pt,
        mark size=2pt,
        ]
    table[x=SNR,y=B16]
    {Data/trellis_b_sweep_4qam_nmse.dat};
    \addplot[ smooth,
              tension=0.1,
             thin,
        color=cssgreen,
        mark=triangle,
        line width=0.75pt,
        mark size=2pt,
        ]
    table[x=SNR,y=B32]
    {Data/trellis_b_sweep_4qam_nmse.dat};
    \addplot[ smooth,
             thin,
             dashed,
        color=chestnut,
        mark=none,
        line width=0.75pt,
        mark size=2pt,
        ]
    table[x=SNR,y=floor_B8]
    {Data/trellis_b_sweep_4qam_nmse.dat};
    \addplot[ smooth,
             thin,
             dashed,
        color=airforceblue,
        mark=none,
        line width=0.75pt,
        mark size=2pt,
        ]
    table[x=SNR,y=floor_B16]
    {Data/trellis_b_sweep_4qam_nmse.dat};
    \addplot[ smooth,
             thin,
             dashed,
        color=cssgreen,
        mark=none,
        line width=0.75pt,
        mark size=2pt,
        ]
    table[x=SNR,y=floor_B32]
    {Data/trellis_b_sweep_4qam_nmse.dat};
\legend{$B=8$,$B=16$,$B=32$,$B=8$ floor,$B=16$ floor,$B=32$ floor};
\end{axis}
\end{tikzpicture}
\vspace{-1em}
\caption{{\color{black}\ac{NMSE} performance of \ac{HiCoMAC} with 4QAM under different
quantization precisions, together with the corresponding
quantization floors. 
}}
\label{fig:trellis_b_sweep_4qam_nmse}
\end{figure}

\begin{figure}[!t]
\centering
\begin{tikzpicture}
    \begin{axis}[
        xlabel = {SNR},
        ylabel = {NMSE},
        label style={font=\footnotesize},
        width=0.37\textwidth,
        height=5cm,
        xmin=3, xmax=15,
        ymin=0.0000000000000000000001, ymax=0.01,
        legend style={nodes={scale=0.65, transform shape}, at={(0.3,0.85)}},
        ticklabel style = {font=\footnotesize},
        legend pos=south west,
        ymajorgrids=true,
        xmajorgrids=true,
        grid style=dashed,
        grid=both,
        ymode = log,
        grid style={line width=.1pt, draw=gray!10},
        major grid style={line width=.2pt,draw=gray!30},
    ]
    \addplot[ smooth,
             thin,
        color=chestnut,
        mark=*,
        line width=0.75pt,
        mark size=1.5pt,
        ]
    table[x=SNR,y=B8]
    {Data/trellis_b_sweep_16qam_nmse.dat};
    \addplot[ smooth,
              tension=0.1,
             thin,
        color=airforceblue,
        mark=asterisk,
        line width=0.75pt,
        mark size=2pt,
        ]
    table[x=SNR,y=B16]
    {Data/trellis_b_sweep_16qam_nmse.dat};
    \addplot[ smooth,
              tension=0.1,
             thin,
        color=cssgreen,
        mark=triangle,
        line width=0.75pt,
        mark size=2pt,
        ]
    table[x=SNR,y=B32]
    {Data/trellis_b_sweep_16qam_nmse.dat};
    \addplot[ smooth,
             thin,
             dashed,
        color=chestnut,
        mark=none,
        line width=0.75pt,
        mark size=2pt,
        ]
    table[x=SNR,y=floor_B8]
    {Data/trellis_b_sweep_16qam_nmse.dat};
    \addplot[ smooth,
             thin,
             dashed,
        color=airforceblue,
        mark=none,
        line width=0.75pt,
        mark size=2pt,
        ]
    table[x=SNR,y=floor_B16]
    {Data/trellis_b_sweep_16qam_nmse.dat};
    \addplot[ smooth,
             thin,
             dashed,
        color=cssgreen,
        mark=none,
        line width=0.75pt,
        mark size=2pt,
        ]
    table[x=SNR,y=floor_B32]
    {Data/trellis_b_sweep_16qam_nmse.dat};
\legend{$B=8$,$B=16$,$B=32$,$B=8$ floor,$B=16$ floor,$B=32$ floor};
\end{axis}
\end{tikzpicture}
\vspace{-1em}
\caption{{\color{black}\ac{NMSE} performance of \ac{HiCoMAC} with 16QAM under different
quantization precisions, together with the corresponding
quantization floors. 
}}
\label{fig:trellis_b_sweep_16qam_nmse}
\end{figure}
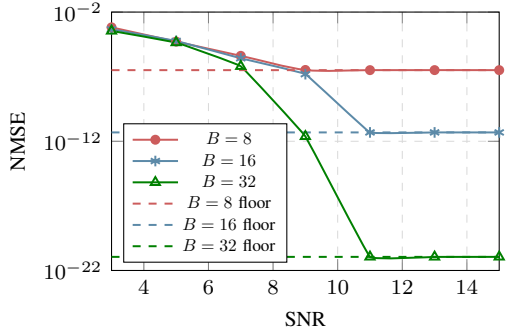

\subsection{Performance of Computational Free Distance Optimization} 

Table~\ref{tab:dcomp_nmse_gain_k2} summarizes the computational
free distance optimization and the resulting NMSE improvement
for different encoder and modulation configurations. The NMSE gain
at SNR $\gamma$ is defined as
\begin{equation}
    G_{\mathrm{NMSE}}(\gamma)
    =
    10\log_{10}
    \left(
    \frac{
        \mathrm{NMSE}_{\rho=1}(\gamma)
    }{
        \mathrm{NMSE}_{\hat{\rho}}(\gamma)
    }
    \right).
    \label{eq:nmse_gain}
\end{equation}
Overall, a larger computational free distance generally leads to improved NMSE performance, since it increases the separation
between computationally distinct paths and reduces the probability of selecting an incorrect sum sequence. The resulting gain depends on the encoder, modulation order, and operating SNR.
When the optimization yields $\hat{\rho}=1$, neither the computational
free distance nor the NMSE is improved. These results show that
$D_{\mathrm{comp}}(\rho)$ provides a useful design criterion across
different coded and modulation configurations.

\vspace{-1em}
\subsection{Performance of System Parameters}

\begin{figure}[!t]
\centering
\begin{tikzpicture}
    \begin{axis}[
        xlabel = {SNR},
        ylabel = {NMSE},
        label style={font=\footnotesize},
        width=0.37\textwidth,
        height=5cm,
        xmin=3, xmax=15,
        ymin=0.0000000000000000000001, ymax=0.1,
        legend style={nodes={scale=0.58, transform shape}, at={(0.3,0.85)}},
        ticklabel style = {font=\footnotesize},
        legend pos=north east,
        ymajorgrids=true,
        xmajorgrids=true,
        grid style=dashed,
        grid=both,
        ymode = log,
        grid style={line width=.1pt, draw=gray!10},
        major grid style={line width=.2pt,draw=gray!30},
    ]
    \addplot[ smooth,
              tension=0.1,
             thin,
        color=chestnut,
        mark=*,
        line width=0.75pt,
        mark size=1.5pt,
        ]
    table[x=SNR,y=R12B16]
    {Data/trellis_rate_bits_m2_nmse.dat};
    \addlegendentry{$R=1/2$, $B=16$}
    \addplot[ smooth,
              tension=0.1,
             thin,
        color=airforceblue,
        mark=asterisk,
        line width=0.75pt,
        mark size=2pt,
        ]
    table[x=SNR,y=R12B32]
    {Data/trellis_rate_bits_m2_nmse.dat};
    \addlegendentry{$R=1/2$, $B=32$}
    \addplot[ smooth,
              tension=0.1,
             thin,
        color=cssgreen,
        mark=triangle,
        line width=0.75pt,
        mark size=2pt,
        ]
    table[x=SNR,y=R14B16]
    {Data/trellis_rate_bits_m2_nmse.dat};
    \addlegendentry{$R=1/4$, $B=16$}
    \addplot[ smooth,
              tension=0.1,
             thin,
        color=cadmiumorange,
        mark=square,
        line width=0.75pt,
        mark size=2pt,
        ]
    table[x=SNR,y=R14B32]
    {Data/trellis_rate_bits_m2_nmse.dat};
    \addlegendentry{$R=1/4$, $B=32$}
\end{axis}
\end{tikzpicture}
\vspace{-1em}
\caption{{\color{black}\ac{NMSE} performance of \ac{HiCoMAC} under different quantization precisions and encoding rates with $K=2$ users. 
}}
\label{fig:trellis_rate_bits_m2_nmse}
\end{figure}

\begin{figure}[!t]
\centering
\begin{tikzpicture}
    \begin{axis}[
        xlabel = {SNR},
        ylabel = {NMSE},
        label style={font=\footnotesize},
        width=0.37\textwidth,
        height=5cm,
        xmin=3, xmax=15,
        ymin=0.0000000000000000000001, ymax=1.8,
        legend style={nodes={scale=0.58, transform shape}, at={(0.3,0.85)}},
        ticklabel style = {font=\footnotesize},
        legend pos=north east,
        ymajorgrids=true,
        xmajorgrids=true,
        grid style=dashed,
        grid=both,
        ymode = log,
        grid style={line width=.1pt, draw=gray!10},
        major grid style={line width=.2pt,draw=gray!30},
    ]
    \addplot[ smooth,
              tension=0.1,
             thin,
        color=chestnut,
        mark=*,
        line width=0.75pt,
        mark size=1.5pt,
        ]
    table[x=SNR,y=R12B16]
    {Data/trellis_rate_bits_m3_nmse.dat};
    \addlegendentry{$R=1/2$, $B=16$}
    \addplot[ smooth,
              tension=0.1,
             thin,
        color=airforceblue,
        mark=asterisk,
        line width=0.75pt,
        mark size=2pt,
        ]
    table[x=SNR,y=R12B32]
    {Data/trellis_rate_bits_m3_nmse.dat};
    \addlegendentry{$R=1/2$, $B=32$}
    \addplot[ smooth,
              tension=0.1,
             thin,
        color=cssgreen,
        mark=triangle,
        line width=0.75pt,
        mark size=2pt,
        ]
    table[x=SNR,y=R14B16]
    {Data/trellis_rate_bits_m3_nmse.dat};
    \addlegendentry{$R=1/4$, $B=16$}
    \addplot[ smooth,
              tension=0.1,
             thin,
        color=cadmiumorange,
        mark=square,
        line width=0.75pt,
        mark size=2pt,
        ]
    table[x=SNR,y=R14B32]
    {Data/trellis_rate_bits_m3_nmse.dat};
    \addlegendentry{$R=1/4$, $B=32$}
\end{axis}
\end{tikzpicture}
\vspace{-1em}
\caption{{\color{black}\ac{NMSE} performance of \ac{HiCoMAC} under different quantization
precisions and encoding rates with $K=2$ users. 
}}
\label{fig:trellis_rate_bits_m3_nmse}
\end{figure}

We then evaluate the effect of the number of quantization bits $B$. Figs.~\ref{fig:trellis_b_sweep_4qam_nmse} and~\ref{fig:trellis_b_sweep_16qam_nmse} show the NMSE performance for the 4\ac{QAM} and 16\ac{QAM} cases with $K=4$ users and memory order $m=2$. The encoder generators are $g=(1,5)_8$ and $g=(1,3,5,7)_8$, and the quantization bits are $B\in\{8,16,32\}$. For both 4QAM and 16QAM, the NMSE decreases as the SNR increases and eventually approaches the computation error floor given in~\eqref{eq:residual_computational_error}, which is the residual NMSE caused by the $B$-bit quantization process. Therefore, a larger $B$ achieves a lower error floor since the inputs are represented with higher precision.
Approximately from
$3$ to $9$ dB for 4\ac{QAM} and from $3$ to $7$ dB for 16\ac{QAM}, the curves with different $B$ are close to each other, where the \ac{NMSE} is mainly caused by incorrect histogram-state trellis decoding decisions. As the SNR increases above $9$ dB, the decoding errors are gradually suppressed. In particular, the $B=8$ curve
saturates first at the highest computation error floor, while the $B=16$ and $B=32$ curves
continue to decrease and approach lower floors, indicating that the quantization precision mainly determines the final \ac{NMSE} once the decoder becomes sufficiently reliable.

Figs.~\ref{fig:trellis_rate_bits_m2_nmse} and
\ref{fig:trellis_rate_bits_m3_nmse} compare the coding rates $R\in\{1/2, 1/4\}$ with $K=2$ users, corresponding to generators $g=(1,3)_8$ and $g=(1,3,5,7)_8$.
Each user input is quantized using $B\in\{16, 32\}$ bits.
For both memory orders $m\in \{1,2\}$, the lower rate configuration $R=1/4$ achieves lower
NMSE over \ac{SNR} approximately from 3 to 12 dB. With the same energy budgets, the higher-order 16QAM mapping
provides more distinct coded modulation points, which can
increase the separation between computationally distinct
aggregated signal paths and thereby improve decoding
reliability.
When the \ac{SNR} further increases, approximately above $12$ dB, the \ac{NMSE} curves approach
their computation error floors,
where the quantization bits determines the final \ac{NMSE}.

\begin{figure}[!t]
\centering
\begin{tikzpicture}
    \begin{axis}[
        xlabel = {SNR},
        ylabel = {NMSE},
        label style={font=\footnotesize},
        width=0.37\textwidth,
        height=5cm,
        xmin=3, xmax=15,
        ymin=0.0000001, ymax=0.005,
        legend style={nodes={scale=0.58, transform shape}, at={(0.3,0.85)}},
        ticklabel style = {font=\footnotesize},
        legend pos=north east,
        ymajorgrids=true,
        xmajorgrids=true,
        grid style=dashed,
        grid=both,
        ymode = log,
        grid style={line width=.1pt, draw=gray!10},
        major grid style={line width=.2pt,draw=gray!30},
    ]
    \addplot[ smooth,
              tension=0.1,
             thin,
        color=chestnut,
        mark=*,
        line width=0.75pt,
        mark size=1.5pt,
        ]
    table[x=SNR,y=K2_NMSE]
    {Data/trellis_K_m2_nmse.dat};
    \addlegendentry{$K=2$}
    \addplot[ smooth,
              tension=0.1,
             thin,
        color=airforceblue,
        mark=asterisk,
        line width=0.75pt,
        mark size=2pt,
        ]
    table[x=SNR,y=K3_NMSE]
    {Data/trellis_K_m2_nmse.dat};
    \addlegendentry{$K=3$}
    \addplot[ smooth,
              tension=0.1,
             thin,
        color=cssgreen,
        mark=triangle,
        line width=0.75pt,
        mark size=2pt,
        ]
    table[x=SNR,y=K4_NMSE]
    {Data/trellis_K_m2_nmse.dat};
    \addlegendentry{$K=4$}
    \addplot[ smooth,
              tension=0.1,
             thin,
        color=cadmiumorange,
        mark=square,
        line width=0.75pt,
        mark size=2pt,
        ]
    table[x=SNR,y=K5_NMSE]
    {Data/trellis_K_m2_nmse.dat};
    \addlegendentry{$K=5$}
\end{axis}
\end{tikzpicture}
\vspace{-1em}
\caption{{\color{black}\ac{NMSE} performance of \ac{HiCoMAC} under different number of users. 
}}
\label{fig:trellis_K_m2_nmse}
\end{figure}
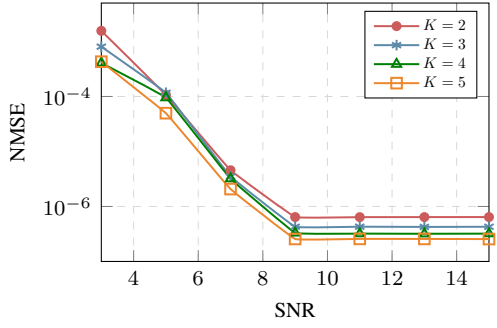

We further evaluate the effect of the number of users in
Fig.~\ref{fig:trellis_K_m2_nmse} for \(K\in\{2,3,4,5\}\), and the encoder uses \(m=2\), \(g=(1,5)_8\), and
\(B=8\) quantization bits. Although minor curve crossings are observed at low SNR, the NMSE
tends to decrease with the number of users as the SNR increases.
When the decoding errors introduced by wireless channels become sufficiently small,
the NMSE is mainly determined by the quantization error and approaches
a floor that decreases approximately as \(1/K\) under the energy
normalization.

\vspace{-1em}
\subsection{Comparison with Existing Digital AirComp Methods}

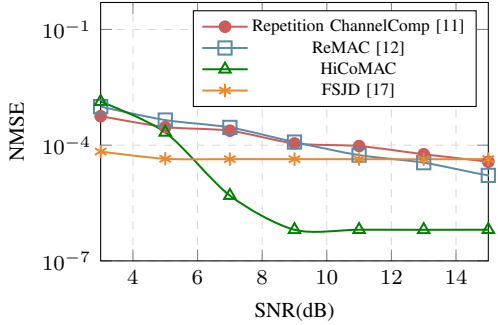
\begin{figure}[!t]
\centering
\begin{tikzpicture}
    \begin{axis}[
        xlabel = {SNR(dB)},
        ylabel = {NMSE},
        label style={font=\footnotesize},
        width=0.37\textwidth,
        height=5cm,
        xmin=3, xmax=15,
        ymin=0.0000001, ymax=0.5,
        legend style={nodes={scale=0.65, transform shape}, at={(0.3,0.85)}},
        ticklabel style = {font=\footnotesize},
        legend pos=north east,
        ymajorgrids=true,
        xmajorgrids=true,
        grid style=dashed,
        grid=both,
        ymode = log,
        grid style={line width=.1pt, draw=gray!10},
        major grid style={line width=.2pt,draw=gray!30},
    ]
    \addplot[ smooth,
             thin,
        color=chestnut,
        mark=*,
        line width=0.75pt,
        mark size=2pt,
        ]
    table[x=SNR,y=Repetition_ChannelComp]
    {Data/comparison_digital.dat};
    \addplot[ smooth,
             thin,
        color=airforceblue,
        mark=square,
        mark options = {rotate = 180, solid},
        line width=0.75pt,
        mark size=2.5pt,
        ]
    table[x=SNR,y=ReMAC]
    {Data/comparison_digital.dat};
    \addplot[ smooth,
             thin,
        color=cssgreen,
        mark=triangle,
        line width=0.75pt,
        mark size=2.5pt,
        ]
    table[x=SNR,y=HiCoMAC]
    {Data/comparison_digital.dat};
    \addplot[ smooth,
             thin,
        color=cadmiumorange,
        mark=asterisk,
        mark options = {rotate = 180, solid},
        line width=0.75pt,
        mark size=2.5pt,
        ]
    table[x=SNR,y=FSJD]
    {Data/comparison_digital.dat};
\legend{Repetition ChannelComp~\cite{saeed2023ChannelComp} ,ReMAC~\cite{yan2025remac} ,HiCoMAC, FSJD~\cite{you2023broadband}};
\end{axis}
\end{tikzpicture}
\vspace{-1em}
\caption{Performance comparison among Repetition ChannelComp, \ac{ReMAC}, and the \ac{HiCoMAC} method for computing the arithmetic-sum function.}
\label{fig:comparison_digital}
\end{figure}

Fig.~\ref{fig:comparison_digital} compares \ac{HiCoMAC} with
Repetition ChannelComp, \ac{ReMAC}, and the \ac{FSJD}
in~\cite{you2023broadband} for \(K=2\) users. All schemes use
10 transmissions per computation with per-user average transmit
energy \(E_{\mathrm s}=1\). Repetition ChannelComp repeats the same
modulated symbol, while \ac{ReMAC} employs selective transmissions.
For \ac{HiCoMAC}, each input is quantized using \(B=8\) bits and
encoded by a \(R=1/2\) convolutional encoder with memory \(m=2\)
and generators \(g=(1,5)_8\). For \ac{FSJD}, each input is represented by
5 information bits, encoded into 10 coded bits and transmitted using
10 \ac{BPSK} symbols.

Compared with Repetition ChannelComp and \ac{ReMAC}, \ac{HiCoMAC} benefits from the temporal redundancy of convolutional coding and
histogram-state trellis decoding, resulting in lower NMSE as the \ac{SNR} increases. Furthermore, under the same transmission resource, \ac{FSJD} uses a 5-bit representation and is therefore more robust at low \ac{SNR}, whereas \ac{HiCoMAC} retains an 8-bit representation. As the SNR increases, the higher quantization precision of \ac{HiCoMAC} becomes dominant, leading to a lower computation error floor. Hence, \ac{HiCoMAC} achieves higher computation accuracy at high \ac{SNR}, while FSJD provides better robustness at low \ac{SNR}.

\vspace{-1em}
\subsection{Complexity Analysis}

\begin{figure}[!t]
\centering
\begin{tikzpicture}
    \begin{axis}[
        xlabel = {Number of users $K$},
        ylabel = {Trellis transitions per stage},
        label style={font=\footnotesize},
        width=0.37\textwidth,
        height=5cm,
        xmin=1, xmax=7,
        ymin=1, ymax=500000000,
        xtick={1,2,3,4,5,6,7},
        ymode = log,
        legend style={nodes={scale=0.62, transform shape}},
        legend pos=north west,
        ticklabel style = {font=\footnotesize},
        ymajorgrids=true,
        xmajorgrids=true,
        grid style=dashed,
        grid=both,
        grid style={line width=.1pt, draw=gray!10},
        major grid style={line width=.2pt,draw=gray!30},
    ]
    \addplot[ smooth,
             thin,
        color=chestnut,
        mark=*,
        line width=0.75pt,
        mark size=1pt,
        ]
    table[x=K,y=hist_edges_m2]
    {Data/trellis_code_decoder_complexity.dat};
    \addlegendentry{\ac{HiCoMAC}, $m=2$}
    \addplot[ smooth,
             thin,
             dashed,
        color=chestnut,
        mark=*,
        mark options = {rotate = 45, solid},
        line width=0.75pt,
        mark size=1pt,
        ]
    table[x=K,y=joint_edges_m2]
    {Data/trellis_code_decoder_complexity.dat};
    \addlegendentry{Joint-user, $m=2$}
    \addplot[ smooth,
             thin,
        color=airforceblue,
        mark=square,
        mark options = {rotate = 45, solid},
        line width=0.75pt,
        mark size=2pt,
        ]
    table[x=K,y=hist_edges_m3]
    {Data/trellis_code_decoder_complexity.dat};
    \addlegendentry{\ac{HiCoMAC}, $m=3$}
    \addplot[ smooth,
             thin,
             dashed,
        color=airforceblue,
        mark=square,
        mark options = {rotate = 45, solid},
        line width=0.75pt,
        mark size=2pt,
        ]
    table[x=K,y=joint_edges_m3]
    {Data/trellis_code_decoder_complexity.dat};
    \addlegendentry{Joint-user, $m=3$}
\end{axis}
\end{tikzpicture}
\vspace{-1em}
\caption{{\color{black}Number of trellis transitions processed per stage by the histogram-state Viterbi decoder of \ac{HiCoMAC} and the ordered
joint-user decoder. 
}}
\label{fig:trellis_code_decoder_complexity}
\end{figure}
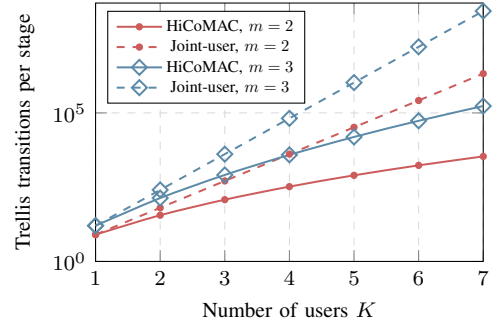

Finally, Fig.~\ref{fig:trellis_code_decoder_complexity} compares the number of
trellis transitions processed per stage by the proposed log-max \ac{MAP} histogram-state
Viterbi decoder and the full-state joint convolutional decoder
in~\cite{you2023broadband}.
For a convolutional encoder with \(N=2^m\) single-user states, the
joint full-state trellis contains \((2N)^K\) transitions
per stage. In contrast, the histogram-state decoder processes $C_{\mathrm H}$ histogram transitions per stage. 
As illustrated in Fig.~\ref{fig:trellis_code_decoder_complexity}, the gap between the two decoders increases as \(K\) grows, particularly
for larger encoder memory. This reduction follows from the permutation
invariance of arithmetic-sum computation, allowing multiple
joint full-state transitions to be represented by a single
histogram-state transition. Consequently, the histogram-state
representation substantially reduces the online branch evaluations by the Viterbi decoder.

\vspace{-1em}
\section{Conclusion}
\label{sec:conclusion}

This paper proposed \ac{HiCoMAC}, a convolutional-coded digital
\ac{AirComp} method for reliable arithmetic-sum computation. By
representing the joint encoder state through the numbers of users
occupying the individual encoder states, HiCoMAC enables efficient log-max \ac{MAP}
histogram-state Viterbi decoding. We further introduced the
computational free distance and developed a modified product graph for
its shortest path evaluation. Based on this distance criterion, the
rectangular ratio of a fixed labeled \ac{QAM} constellation was
optimized through a coarse-to-fine search. Simulation
results demonstrated that \ac{HiCoMAC} improves computation reliability in
terms of both computational free distance and \ac{NMSE}, while
substantially reducing the decoding complexity. Extending
the method to imperfect channel compensation, heterogeneous
encoders, while reducing the offline product graph complexity, constitutes future work.

\appendix




\subsection{Proof of Proposition~\ref{prop:histogram_edge_cardinality}}
\label{proof:histogram_edge_cardinality}
For a fixed source histogram state
\(\bm n^{(i)}\), each component can take
\(n_\sigma^{(i)}+1\) possible values. Since each feasible
input-count vector \(a_{ij,\sigma}^{r}\) uniquely determines the resulting target histogram state through the encoder transition rule, the number of outgoing histogram edges from \(v_i\) is $\prod\nolimits_{\sigma=1}^{N} (n_\sigma^{(i)}+1)$.
Summing over all possible histogram states yields
\begin{equation}
\label{eq:histogram edge summation} 
    C_{\mathrm{ H}}=
    \sum\nolimits_{\substack{
    n_1+\cdots+n_N=K\\
    n_\sigma\geq0}}
    \prod\nolimits_{\sigma=1}^{N}(n_\sigma+1).
\end{equation}
The generating function associated with each factor is
\begin{equation}
    \sum\nolimits_{n=0}^{\infty}(n+1)x^n
    =
    (1-x)^{-2}.
\end{equation}
Therefore, the sum in \eqref{eq:histogram edge summation} is the coefficient of \(x^K\)
in \((1-x)^{-2N}\). By the generalized binomial expansion, we obtain
\begin{equation}
    (1-x)^{-2N}
    =
    \sum\nolimits_{r=0}^{\infty}
    \binom{r+2N-1}{2N-1}x^r.
\end{equation}
Taking the coefficient of \(x^K\) gives
\begin{equation}
    C_{\mathrm H}
    =
    \binom{K+2N-1}{2N-1},
\end{equation}
which completes the proof.

\vspace{-1em}
\subsection{Proof of
Theorem~\ref{thm:shortest_path_characterization}}
\label{app:shortest_path_characterization}

Consider an irreducible computational error event
\((\pi,\widetilde{\pi})\) and its induced modified product path $\omega \in\mathcal P_{\mathrm M}$.
Let \(\zeta^{(h)}\in\mathcal C_{\varepsilon^{(h)}}\)
be the product edge corresponding to the \(h\)-th paired histogram
transition of the event. We then have
\begin{align} \nonumber
    D(\pi,\widetilde{\pi};\rho)
    \geq
    \sum\nolimits_{h=1}^{\nu}
    \min_{\zeta\in\mathcal C_{\varepsilon^{(h)}}}
    D_{\zeta}(\rho) 
    =
    W_{\mathrm M}(\omega;\rho).
    \label{eq:event_to_modified_path_bound}
\end{align}

Since every irreducible computational error event corresponds to a path in
\(\mathcal P_{\mathrm M}\), it follows that
\begin{equation}
    D_{\mathrm{comp}}(\rho)
    \geq
    \min_{\omega\in\mathcal P_{\mathrm M}}
    W_{\mathrm M}(\omega;\rho).
    \label{eq:shortest_path_lower_bound}
\end{equation}

Conversely, consider any
\(\omega
=(\varepsilon^{(1)},\ldots,\varepsilon^{(\nu)})
\in\mathcal P_{\mathrm M}\).
For each \(h\), select a product edge satisfying
\begin{equation}
    \zeta^{(h)}
    \in
    \arg\min_{\zeta\in\mathcal C_{\varepsilon^{(h)}}}
    D_{\zeta}(\rho).
\end{equation}

Since all candidates in
\(\mathcal C_{\varepsilon^{(h)}}\) have the same endpoints in the
modified product graph, the minimum distance associated with a fixed
\(\omega\) is calculated as
\begin{equation}
    \sum\nolimits_{h=1}^{\nu}
    D_{\zeta^{(h)}}(\rho)
    =
    W_{\mathrm M}(\omega;\rho).
\end{equation}
Minimizing this quantity over all
\(\omega\in\mathcal P_{\mathrm M}\) gives
\begin{equation}
    D_{\mathrm{comp}}(\rho)
    \leq
    \min_{\omega\in\mathcal P_{\mathrm M}}
    W_{\mathrm M}(\omega;\rho).
    \label{eq:shortest_path_upper_bound}
\end{equation}
Therefore, combining
\eqref{eq:shortest_path_lower_bound} and
\eqref{eq:shortest_path_upper_bound} leads to
\eqref{eq:shortest_path_characterization}, which concludes the proof.

\bibliographystyle{ieeetr}
\bibliography{main.bbl}

\end{document}